\documentclass[a4paper,UKenglish,cleveref, autoref, thm-restate]{lipics-v2021}

\usepackage{preamble}

\title{ Symmetric Proofs in the Ideal Proof System} 

\titlerunning{     }  

\author{Anuj Dawar}{Department of Computer Science and Technology, University of Cambridge, UK}{anuj.dawar@cl.cam.ac.uk}{}{}

\author{Erich Grädel}{Mathematical Foundations of Computer Science, RWTH Aachen University, Germany}{graedel@informatik.rwth-aachen.de}{}{}

\author{Leon Kullmann}{RWTH Aachen University, Germany}{leon.kullmann@rwth-aachen.de}{}{}

\author{Benedikt Pago}{Department of Computer Science and Technology, University of Cambridge, UK}{btp26@cam.ac.uk}{}{}


\authorrunning{A. Dawar, E. Grädel, L. Kullmann, B. Pago} 

\Copyright{Anuj Dawar, Erich Grädel, Leon Kullmann, Benedikt Pago} 

\newcommand{\IPSlin}{\text{\upshape IPS}_{\mathsf{LIN}}}
\newcommand{\symIPS}{\text{ \upshape sym-IPS}}
\newcommand{\symIPSlin}{\text{ \upshape sym-IPS}_{\mathsf{LIN}}}
\newcommand{\VNP}{\mathsf{VNP}}
\newcommand{\VP}{\mathsf{VP}}

\DeclareMathOperator{\charact}{char}

\nolinenumbers

\begin{document}
	\maketitle
	
	\begin{abstract}
          We consider the Ideal Proof System (IPS) introduced by Grochow and Pitassi and 
          pose the question of which tautologies admit symmetric proofs, and of what complexity. The symmetry requirement in proofs is inspired by recent work establishing lower bounds in other symmetric models of computation.
		We link the existence of symmetric IPS proofs to the expressive power of logics such as fixed-point logic with counting and Choiceless Polynomial Time, specifically regarding the graph isomorphism problem.  We identify relationships and tradeoffs between the symmetry of proofs and other parameters of IPS proofs such as size, degree and linearity.  We study these on a number of standard families of tautologies from proof complexity and finite model theory such as the pigeonhole principle, the subset sum problem and the Cai-Fürer-Immerman graphs, exhibiting non-trivial upper bounds on the size of symmetric IPS proofs.
	\end{abstract}

	\section{Introduction}
	A central project within the subject of proof complexity \cite{pitassi, segerlind, krajicekBook} is to define ever more powerful proof systems for propositional logic for which it is possible to establish superpolyonomial lower bounds on the size of proofs. The \emph{Ideal Proof System} (IPS), introduced by Grochow and Pitassi in~\cite{IPS}, is a powerful \emph{algebraic} proof system that subsumes many previous algebraic and propositional proof systems such as the polynomial calculus and Frege proof systems. No superpolynomial lower bounds are known for the unrestricted IPS and it has been shown, for instance, that it efficiently simulates powerful proof systems such as extended Frege, for which such lower bounds are a long-standing open problem.  On the other hand, lower bounds have been shown for interesting restrictions of IPS.
	
	In IPS, like other algebraic proof systems, propositional formulas are coded as polynomial equations and a proof (or refutation) is a certificate showing that a given equation system has no solution.  Equivalently, the certificate shows that the input polynomials, or \emph{axioms}, have no common zero. In IPS, the certificate is an \emph{algebraic circuit} which witnesses that the unit polynomial $1$ is in the ideal generated by the axioms, and we measure the complexity of the certificate in terms of the size of the circuit.  Here, an algebraic circuit is a representation of a polynomial as a circuit with constants and variables as inputs and addition and multiplication gates. In order to make progress on showing lower bounds for IPS, it has been suggested~\cite{IPS} to study fragments of IPS given by restrictions on the circuits considered.
	A successful recent line of work has established super-polynomial lower bounds for fragments such as bounded-depth circuits, multilinear formulas and so-called read-once oblivious algebraic branching programs (roABPs) \cite{lowerbounds_IPS, limaye2021new, govindasamy2022simple, hakoniemi2024functional}. The common method behind these results has been named the \emph{functional lower bound} technique.
	
	In the present paper, we consider a natural restriction on IPS proofs based on \emph{symmetric circuits}. Non-trivial lower bounds for symmetric algebraic circuits have recently been established~\cite{DawarW20,DawarW22,DawarPS25} and this raises the question of whether, when the input polynomials have natural symmetries, the corresponding proof can also be made symmetric.  Or, perhaps more interestingly, whether the lower bound methods for symmetric algebraic circuits can be combined with functional lower bound techniques to obtain lower bounds for symmetric IPS proofs.
	
	The present paper lays the foundations for this. We formalise the symmetry requirement in IPS proofs and show that our definition is sound in the sense that restricting to symmetric proofs still yields a complete proof system. Beyond this, our work has two main contributions. First, we show how the study of symmetric IPS fits into the context of a larger project that investigates the role of symmetry in lower bounds in a number of areas of complexity theory (see~\cite{DawarCSL20}). This project is centred around the concept of \emph{symmetric computation} in the sense of symmetric circuits and logics from finite model theory~\cite{DawarICALP24,DawarPago24}. 
	Specifically, it seeks to lift game-based methods from finite model theory which are used to show inexpressiveness results in logic to lower bounds in other models of computing. 
	There is a well-established connection between the expressive power of such logics and propositional proof complexity. In particular, it has been shown that, by encoding propositional axioms as relational structures in a suitable way, the problem of the existence of a resolution refutation of a given width, or of a polynomial calculus refutation of a fixed degree are expressible in existential fixed-point logic and fixed-point logic with counting (FPC) respectively and are indeed \emph{complete} for these logics under weak, symmetry-preserving reductions~\cite{ggpp}.  Hence, in a precise sense these logics exactly characterise the power of the corresponding proof systems and from lower bounds on the expressive power of the logics we can extract lower bounds in proof complexity. Moreover, the stronger logic Choiceless Polynomial Time (CPT)  can, in a symmetry-preserving way, be simulated in the bounded-degree extended polynomial calculus~\cite{Pago}. CPT is of great importance in descriptive complexity as a candidate logic for capturing PTIME \cite{pakusa2015linear, gradel2015polynomial}.
	
	As the results we show in Section~\ref{sec:appplications} establish, symmetric IPS proofs offer a suitable unifying formalism in which the results from~\cite{ggpp, Pago} can be recast. The expressive power of FPC corresponds to the bounded-degree symmetric IPS proofs, and the stronger logic CPT can only be simulated by symmetric IPS proofs of unbounded degree.   Importantly, the symmetries of the IPS proofs we consider are the inherent symmetries of the input structure, which are the symmetries under which the corresponding FPC or CPT computation is invariant.
	
	The second focus of this paper is upper bounds for the size of symmetric IPS refutations of various families of instances that possess natural symmetries. 
	One example, where upper (and lower) bounds actually follow from results in finite model theory, is a formulation of the \emph{graph isomorphism problem} as a system of polynomial equations \cite{GI} over $\bbQ$. The equations expressing that two graphs $G$ and $H$ are isomorphic are symmetric under the automorphism groups of $G$ and $H$. We show that the complexity of symmetric IPS refutations with these symmetry groups matches the descriptive complexity of distinguishing $G$ and $H$: Graphs that are distinguishable in FPC, or equivalently by the $k$-dimensional \emph{Weisfeiler-Leman algorithm} \cite{sandraSiglog} for fixed $k$, admit bounded-degree polynomial-size symmetric IPS refutations. Graphs that are distinguishable in the stronger logic CPT, but not in FPC, admit polynomial-size symmetric IPS refutations, but not of bounded degree. This separates in particular the bounded-degree fragment of symmetric IPS from its unbounded-degree version (Theorem \ref{thm:separationOfSymBoundedDegree}).
	
	A related family of instances we consider is a system of linear equations over the finite field $\bbF_2$ that expresses the isomorphism problem on Cai-Fürer-Immerman (CFI) graphs~\cite{caifurimm92, atseriasBulatovDawar} which is known to be inexpressible in FPC. Even though the instances exhibit a large number of interesting symmetries and are the core of numerous finite model theory lower bounds, it turns out quite surprisingly that they do admit polynomial-size symmetric IPS refutations. 
	However, we are not able to exhibit a \emph{linear} symmetric IPS refutation of the CFI equations of less than exponential size, making them a candidate for separating the linear fragment from the full symmetric IPS.  Here the linear (or \emph{Hilbert-like}) fragment is as defined in~\cite{IPS} (see also Section~\ref{sec:IPS}).
	
	Besides the CFI equations, we consider standard hard examples from proof complexity, namely the \emph{pigeonhole principle} (PHP) and the \emph{subset sum} problem. The former is the classic tautology asserting that there is no injective function from a set of $n+1$ elements to a set of $n$ elements. The latter is simply the equation $\sum_{i=1}^n x_i - \beta = 0$, together with the equations $x_i^2 - x_i = 0$ for all $i$, enforcing $\{0,1\}$-solution. This is unsatisfiable when $\beta > n$.  The pigeonhole principle has been used to obtain lower bounds for bounded-width resolution~\cite{haken} and bounded-degree polynomial calculus~\cite{razborov1998lower}.  This is again a promising candidate for showing lower bounds on the size of symmetric IPS proofs as the tautologies are rich in symmetries and we know of no sub-exponential size symmetric IPS refutations for PHP.
	
	In contrast, for subset sum, we do obtain polynomial size symmetric IPS refutations, which is significant because subset sum and its variations are the key instances where the functional lower bound method \cite{lowerbounds_IPS, govindasamy2022simple, hakoniemi2024functional} has been used successfully. Thus, symmetry alone is a limitation of a rather different nature than others, e.g.\ multilinear-formula or roABP-IPS where functional lower bounds have been established. 
	
	Our upper bounds on proof sizes are summarised in the table below. In this table, the entry ``none'' indicates that no proofs with the given restrictions exist, and $c$ is a fixed constant. Our upper bounds for the graph isomorphism problem depend on the graphs: 
	We consider graphs that are distinguishable by the $k$-dimensional Weisfeiler Leman (WL) algorithm versus graphs that are distinguished by CPT; the latter distinguishes strictly more graphs \cite{dawar2008}.
	\begin{table}[h]
		\centering
		\begin{tabular}{c|c|c|c|c|c}
			&\makecell{\textbf{FPC-definable}\\ \textbf{problems}} &\makecell{\textbf{Graph}\\ \textbf{isomorphism}} & \textbf{CFI} & \textbf{Subset sum} & \makecell{\textbf{Pigeonhole}\\ \textbf{principle}}\\
			\hline
			$\deg_k$-sym-IPS & $\Oo(n^c)$ & \makecell{$\Oo(n^c)$ if $k$-WL-\\ distinguishable} & none & none & none\\
			\hline
			\makecell{$\symIPSlin$} & $\Oo(n^c)$ & \makecell{$\Oo(n^c)$ if CPT-\\ distinguishable}& $\Oo(2^n)$ & $\Oo(n^c)$ & $\Oo(3^n \cdot n)$ \\
			\hline
			sym-IPS & $\Oo(n^c)$ & \makecell{$\Oo(n^c)$ if CPT-\\ distinguishable} & $\Oo(n^c)$ & $\Oo(n^c)$ & $\Oo(3^n \cdot n)$
		\end{tabular}
	\end{table}	

	\paragraph*{Acknowledgments} We are grateful to Iddo Tzameret for familiarising us with the current trends in IPS lower bounds and for valuable suggestions for future directions of this project.

	\section{Preliminaries}\label{sec:preliminaries}
	We denote by $[n]$ the set $\{1,...,n\}$.  We call a set $X$, with the action of a group $\Gamma$ on it a \emph{$\Gamma$-set}.  The action of $\Gamma$ on $X$ has a \emph{natural extension} to a number of other sets defined from $X$.  In particular, the following are $\Gamma$-sets: $X^k$ for any $k \in \bbN$; the set $A^X$ of functions from $X$ to any set $A$; and the collection $\bbF[X]$ of polynomials over $X$ with coefficients from a field $\bbF$.
	
	If $X$ is a $\Gamma$-set and $p \in \bbF[X]$ a polynomial, then we say that $p$ is \emph{$\Gamma$-symmetric} or \emph{$\Gamma$-invariant} if $\pi(p) = p$ for every $\pi \in \Gamma$.  A set of polynomials $\F \subseteq \bbF[X]$ is $\Gamma$-invariant if $\{ \pi(f) \mid f \in \F  \} = \F$ for every $\pi \in \Gamma$.
	We write $\Sym(X)$ for the \emph{symmetric group} on $X$ and for a finite relational structure $\AA$ we write $\Aut(\AA)$ for the \emph{automorphism group} of $\AA$.
	
	In what follows, $X$ denotes a finite set of variables, $\Field$ denotes a field, and $\Field[X]$ is the corresponding ring of polynomials in the variables $X$. 
	Instances for the algebraic proof systems we consider are \emph{systems of polynomial equations}. A system of polynomial equations is a finite set $\Ff \subseteq \Field[X]$, and a \emph{solution} for $\Ff$ is an assignment $s: X \rightarrow \Field$ such that $f(s(\vec x)) = 0$ for each $f(\vec x) \in \Ff$. We say that $\Ff$ is satisfiable if it has a solution, and unsatisfiable otherwise. If we wish to encode Boolean satisfiability problems, where the assignments are in $\{0,1\}$ only, then we include in $\Ff$ the \emph{Boolean axioms} $x^2-x$ for each $x \in X$. 	
	 
	For a system of polynomial equations $\F$, the \emph{degree} $\deg(\F)$ of $\F$ is the maximum degree of any polynomial in $\F$.  The \emph{size} $|\F|$ of $\F$ is the total number of variables and equations, i.e.\ if $\F = \defset{f_j\in \Field[X]}{1 \leq j \leq m}$, then $|\F| \coloneqq m + |X|$.

	\subsection{(Symmetric) Algebraic Circuits}
	\label{sec:circuits}
	We study the complexity of a polynomial not only in terms of its degree or its number of monomials but also via the size of the smallest algebraic circuit representing it. 
	\begin{definition}[Algebraic circuit]
		An \textit{algebraic circuit} over variables $X$ and a field $\bbF$ is
		a connected directed acyclic graph $\Circ = (D,W)$ with a labelling
		$\lambda: D \rightarrow \{+,\times \} \cup X \cup \bbF$ 
		such that $\lambda(g) \in X \cup \bbF$ if $g$ has in-degree $0$, and $\lambda(g) \in \{+,\times\}$, otherwise.
	\end{definition}
	The nodes in $D$ are called \textit{gates} and the edges in $W$ are called \textit{wires}.
	Gates with in-degree $0$ are called \textit{input gates} and gates with out-degree $0$
	are called \textit{output gates}. Unless stated otherwise,
	there is only one output gate.
	Every gate that is not an input gate is an \textit{internal gate}.
	The \textit{size of a circuit} $\Circ$ is
	$|\Circ| := |D(\Circ)|$.

	We view an algebraic circuit $\Circ$ over $X$ and $\bbF$ as computing
	a polynomial $\Circ(\vec x) \in \bbF[X]$ by evaluating gates to polynomials and treating
	$+, \times$ as ring operations. We write $C$ for the circuit and $C(\vec x)$ for the polynomial it computes. For a gate $g \in D$, we denote by $C_g(\vec x)$ the polynomial computed by the subcircuit rooted in $g$.

	The \textit{semantic degree} of an algebraic circuit is the maximum degree of a polynomial computed by any of its subcircuits:
	$\deg (\Circ) := \max(\defset{\deg(\Circ_g(\vec x))}{g \in D}).$
	
	
	We are mostly concerned with \emph{symmetric} algebraic circuits, as they are studied for example in \cite{DawarW20}. If $\Gamma$ is a group acting on a variable set $X$, then a circuit $C$ over $X$ is \emph{$\Gamma$-symmetric} if we can extend the action of $\Gamma$ on the input gates
	to automorphisms of the whole circuit.  An automorphism of a circuit is a permutation of the gates that preserves its structure as a DAG, as well as all labels.
	Formally, $\pi \in \Gamma$ \emph{extends} to an automorphism of $C = (D, W)$ if there is a $\sigma \in \Sym(D)$ such that $\sigma$ is an automorphism of the DAG $C$ that preserves labels of internal gates, and for every input gate $g \in D$ with label $\lambda(g) \in X \cup \bbF$, we have $\lambda(\sigma(g)) = \pi(\lambda(g))$. 
	Here, the action of $\Gamma$ on $\bbF$ is trivial, i.e.\ if $\lambda(g) \in \bbF$, then $\pi(\lambda(g)) = \lambda(g)$.
	It is easy to check that the polynomial computed by a $\Gamma$-symmetric algebraic circuit is itself invariant under the action of $\Gamma$ on its variables.

	\subsection{Logics from finite model theory}
	\label{sec:logics}
	\paragraph*{Fixed-point logic with counting}
	
	Formulas of fixed-point logic with counting (FPC) are evaluated in finite structures expanded with a numeric sort as the domain for counting terms. 
	For every finite structure $\AA$, let $\AA^* := \AA \uplus (\bbN, <, +, \cdot, 0, e)$, where $e$ is interpreted as $|A|$. The variables in FPC are typed so that each variable either has as its domain the point sort (i.e.\ $\AA$) or the numeric sort of $\AA^*$. Point variables are denoted with Roman letters and numeric variables with Greek letters.
	
	The syntax of FPC is that of first-order logic with the following extensions:
	\begin{itemize}
		\item Quantifiers over numeric variables: If $\mu$ is a numeric variable, then quantification over $\mu$ is only allowed in the form $Q \mu < t. \phi$, where $Q \in \{\exists, \forall\}$ and $t$ is a numeric term. This ensures that every numeric variable has a fixed range of polynomial size (with respect to the size of the point sort) -- otherwise, it would not be possible to evaluate FPC-formulas in polynomial time.
		\item \emph{Counting quantifiers}: If $\phi(\vec x, \vec \mu)$ is an FPC-formula, and $\nu$ a numeric variable, then $\exists^{\geq \nu}\phi(\vec x \vec \mu)$ is a formula which is true in $\AA^*$ if the number of satisfying assignments to $x$ is at least the value of $\nu$.
		\item \emph{Fixed-point operators}: Let $Z$ be a second-order variable, $\phi$ an FPC-formula in which $Z$ occurs only positively, and $t$ a numeric term. Then $[\mathbf{lfp}~ Z \vec x \vec \mu_{<t} . \phi(Z, \vec x, \vec \mu)](\vec x \vec \mu)$ is a formula of FPC. It is satisfied in $(\AA^*,\vec x \vec \mu \mapsto \vec a \vec b)$ if $\vec a \vec b$ is in the least fixed-point of the sequence defined by $Z_0 = \emptyset, Z_{i+1} = \{ \vec c \vec d  \mid \AA^* \models  \phi(Z_i, \vec c, \vec d)  \}$.
		
		There also exists the dual operator 
		$[\mathbf{gfp}~ Z \vec x \vec \mu_{<t} . \phi(Z,\vec x, \vec \mu)](\vec x, \vec \mu)$, which computes the greatest fixed-point.
	\end{itemize}	 
	
	This is one way to present FPC -- there are other equivalent presentations, for example where the numeric sort in $\AA^*$ is finite, or where counting \emph{terms} instead of counting \emph{quantifiers} are used. For more details and background on FPC, we refer to \cite{otto1997bounded} or \cite{dawar2015nature}.

	\paragraph*{Choiceless Polynomial Time}
	Choiceless Polynomial Time (with counting) can be viewed as an extension of FPC with higher-order data structures.  It properly extends the expressive power of FPC and it remains an open question whether it can express all polynomial-time decidable properties of finite structures.  It seeks to overcome the limitation of FPC that the stages of a fixed-point computation are relations of a fixed-arity. In CPT, the computation stages are much more expressive objects, namely \emph{hereditarily finite sets}. These are arbitrarily nested finite sets whose atoms are elements of the respective input structure.   CPT has an iteration mechanism with definable state-updates similar to FPC, with the difference that in each step, new hereditarily finite sets may be constructed. Thus, computations are not guaranteed to reach a fixed-point. To guarantee polynomial time evaluation, every CPT-sentence $\Pi$ comes with an explicit polynomial bound $p(n) \colon \bbN \to \bbN$. On an input structure $\AA$, the evaluation of $\Pi$ is aborted if it takes longer than $p(|\AA|)$ many steps. Likewise, the constructed h.f.\ sets are required to have size at most $p(|\AA|)$. 
	In this article, we write $\CPT(p(n))$ for the fragment of CPT consisting of sentences whose explicit polynomial bound is $p(n)$.
	Since CPT is a logic and the state-updates are definable, the h.f.\ sets constructed during a computation are naturally symmetric under the automorphisms of the input structure.

	We refrain from giving a formal definition of CPT, since the details are not needed here and it would be quite lengthy. 
	A concise survey can be found in \cite{gradel2015polynomial}. The work that originally introduced CPT as an abstract state machine model is \cite{CPT}; for a more ``logic-like'' presentation of CPT, see \emph{BGS-logic} \cite{rossman2010choiceless}.

	\section{Algebraic Proof Systems}
	\label{sec:proofSystems}
	Algebraic proof systems are formalisms for refuting the satisfiability of polynomial equation systems. 
	A refutation is a certificate of unsatisfiability checkable in deterministic or randomized polynomial time. 
	Usually, the certificate is based on Hilbert's Nullstellensatz: It states that for any polynomial system of equations over a field $\Field$,
	$\F \subseteq \Field[X]$, $\F$ is unsatisfiable over the algebraic closure of $\bbF$ if, and only if,
	$1$ is in the ideal generated by $\F$, which means that there are
	$g_1,g_2,...,g_m \in \Field[X]$ such that $\sum_{i \leq m}f_ig_i = 1$. 
	Thus, a refutation of $\F$ in an algebraic proof system is typically a systematic proof of the existence of such $g_1,g_2,...,g_m$.
	The proof systems we are concerned with are (variants of) the \emph{polynomial calculus} (PC) and the \emph{Ideal Proof System} (IPS).
	In the following, whenever we speak about the unsatisfiability of a polynomial system of equations (defined over a field $\bbF$), we implicitly mean its unsatisfiability over the algebraic closure of $\bbF$, as required by the Nullstellensatz. 
	
	The efficiency of different proof systems is compared using the standard notion of \emph{$p$-simulation}:
	Let $P_1,P_2$ be two proof systems and $\mu_1, \mu_2$ be complexity measures for them, that is, functions that map refutations to natural numbers.
	Then we write $P_2 \leq_p P_1$ (with respect to the measures $\mu_1, \mu_2$, which will usually be clear from the context) if for every system of polynomial equations $\Ff$, and every $P_2$-refutation $R_2$, there exists a $P_1$-refutation $R_1$ of $\F$ with $\mu_1(R_1) \leq \poly(\mu_2(R_2))$.
	We write $P_1 \equiv_p P_2$ if $P_2 \leq_p P_1$ and $P_1 \leq_p P_2$.
		
	\subsection{Polynomial Calculus}
	The polynomial calculus \cite{groebner} predates the IPS as an algebraic proof system. Unlike the IPS, the PC is a more classical rule-based system consisting in a set of sound inference rules for systematically deriving the polynomials in the ideal generated by the input $\F$. 
	\begin{definition}[Polynomial Calculus (PC)]
		Let $\F \subseteq \Field[X]$ be a system of polynomial equations over a field $\Field$
		with variables in $X$. \\
		The \textit{inference rules} of the polynomial calculus are:
		\begin{itemize}
			\item Axiom: {\Large$\frac{}{~f~}$} for all $f \in \F$.
			\item Multiplication: {\Large$\frac{f}{~xf~}$} for any $f\in \Field[X], x\in X$.
			\item Linear Combination: {\Large$\frac{f\quad g}{~af ~+~ bg~}$}
			for any $f,g \in \Field[X], a,b \in \Field$.
		\end{itemize}
		A \textit{$\PC$ refutation} of $\F$ is a sequence $(p_1,p_2 , ... ,p_n = 1)$ of polynomials such that $p_n$ is the $1$-polynomial, and each $p_i$ is either a Boolean axiom, an axiom from $\F$, or is the result of the application of a proof rule to one or two polynomials $p_j, p_{j'}$ with $j,j' < i$.
	\end{definition}
	The PC  is a sound and complete proof system, that is, a polynomial equation system $\F$ is unsatisfiable if, and only if, there exists a refutation for $\F$. The complexity of this refutation may in general be super-polynomial.
	
	We consider two complexity measures for the PC. The \textit{number of lines} of a PC refutation
	$R = (p_1,p_2 , ... ,p_n = 1)$ is $n$, i.e.\ the length of the derivation sequence. 
	The \textit{degree} $\deg(R)$ of $R$ is the maximum degree of any of the $p_j$, for $1 \leq j \leq n$.
	For $k \in \bbN$, we denote by $\deg_k$-PC the restriction of the PC where only refutations of degree $\leq k$ are allowed. This is no longer a complete proof system; for instance, it cannot prove the pigeon hole principle \cite{razborov1998lower}.

	\subsection{Ideal Proof System}
	\label{sec:IPS}
	The Ideal Proof System (IPS) introduced by Grochow and Pitassi~\cite{IPS} is a more general formalism for proving in a verifiable way that $1$ is in the ideal generated by a set $\F$ of axioms. In this proof system, there exist no derivation rules, just a certificate. This takes the form of an algebraic circuit. IPS in particular $p$-simulates the polynomial calculus \cite[Proposition 3.4]{IPS}.
	\begin{definition}[Ideal Proof System (IPS), \cite{IPS}]
		Let $\F$ be the polynomial equation system $\{f_1(\vec x), f_2(\vec x), \ldots,f_m(\vec x)\}$
		in $\Field[X]$.
		Let $Y = \{y_1, y_2,...,y_m\}$ be a fresh set of variables.
		An \emph{IPS certificate} of unsatisfiability of $\F$ over $\Field$ is a polynomial
		$C(\vec x, \vec y) \in \Field[X \uplus Y]$ such that $(1) \ C(\vec x, \vec 0 ) = 0$, and $(2) \ C(\vec x, \vec f ) = 1$.\\
		An \emph{IPS proof} of unsatisfiability (i.e.\ a \emph{refutation}) of $\F$ is an algebraic circuit $\Circ$ with variables $X \uplus Y$ and constants $\bbF$ computing an $\IPS$-certificate of unsatisfiability.
	\end{definition}
	Condition (1) ensures that $C(\vec x, \vec f)$ is in the ideal generated by $\F$.
	The unrestricted IPS is sound and complete, i.e.\ there is an  IPS refutation of $\F$ if, and only if, $\F$ is unsatisfiable \cite{IPS}.  Note that this is not necessarily a proof system in the sense of Cook and Reckhow as it is not clear that certificates can be verified in polynomial time.   Verifying that a given circuit indeed computes a valid certificate requires polynomial identity testing (PIT), which is in randomized polynomial time (but not known to be in P).
	\\
	We define the \emph{size} $|C|$ of an IPS proof $C = (D,W)$ as 
	$
	|C| := \min(|D|, |X| + |Y|).
	$
	That is, we define the size of a refutation to be at least the instance size. 
	This would be inadequate for sublinear size refutations but in the symmetric setting that we study here, the smallest possible proof size is generally $|X| + |Y|$.
	
	For families $(\Ff_n)_{(n \in \bbN)}$ of instances and corresponding refutations $(C_n)_{(n \in \bbN)}$, we are mainly interested in IPS refutations of \emph{polynomial size} $|C_n| \leq p(|\F_n|)$.
	When we speak of the complexity of IPS proofs, we usually mean this complexity measure.
	
	For every $k \in \bbN$, $\deg_k$-IPS denotes the restriction of the IPS 
	in which the semantic degree $\deg(C)$ (see Section \ref{sec:circuits}) of any refutation $C$ is at most $k$. 
	
	Other natural restrictions of the IPS are obtained by allowing only circuits from certain circuit classes, such as bounded depth, or as in our case, symmetric circuits.
	By the Nullstellensatz, there always exists a \textit{$\vec y$-linear} (also called 
	\textit{Hilbert-like} in \cite{IPS}) IPS refutation
	$C(\vec x, \vec y)$ for every unsatisfiable $\F$, i.e.\ 
	$
	C(\vec x, \vec y) = \sum_{i=1}^m y_i g_i(\vec x)
	$
	for some $\vec g \in \Field[X]$. In other words, the $\vec y$-linear IPS is a sound and complete fragment of the general IPS. We denote it $\IPSlin$.
	It is shown in \cite{IPS} that $\IPSlin$ $p$-simulates the general IPS on instances that are themselves representable with polynomial-size algebraic circuits. 
	For fragments of the IPS with restricted circuit classes, it is however not clear that every proof can efficiently be simulated by a $\vec y$-linear one. Indeed, for symmetric proofs, as we show, $\IPSlin$ turns out to be a true restriction. 

The unrestricted IPS is remarkably powerful: It is shown in \cite[Theorem 3.5]{IPS} that IPS over any field of characteristic $q$ $p$-simulates any Frege proof system with MOD\subsc{$q$}-connectives. This is even true for Extended Frege, where extension axioms may be used as in the Extended Polynomial Calculus. Frege systems are standard textbook propositional proof systems such as the sequent calculus.

\section{Symmetric IPS Proofs}
\label{sec:symIPS}
We now introduce the symmetry restriction on IPS proofs.  That is to say, we consider when a set $\F$ of polynomials that is $\Gamma$-invariant for a group $\Gamma$ acting on its variables, admits a $\Gamma$-symmetric circuit as an IPS refutation.  We start with a formal definition.

\begin{definition}[Symmetric IPS]
	Let $\Gamma$ be a group and $X$ a $\Gamma$-set of variables. Let $\F = \{ f_1(\vec x), ..., f_m(\vec x)\}  \subseteq \Field[X]$ be a $\Gamma$-invariant set of polynomials and let $Y = \{y_1,...,y_m\}$ be a $\Gamma$-set of variables with the following action: For every $\pi \in \Gamma$ and $ i \in [m]$, $\pi(y_i) = y_j$ for the $j$ with $\pi(f_i) = f_j$.\\
	A \emph{$\Gamma$-symmetric IPS proof} of unsatisfiability of $\F$ is a $\Gamma$-symmetric algebraic circuit with variables $X \uplus Y$
	computing an $\IPS$ certificate $C(\vec x, \vec y)$ of $\F$.
\end{definition}
Note that if $\Gamma = \{\id\}$, then any IPS refutation of $\F$ is also a $\Gamma$-symmetric IPS refutation of $\F$. The complexity of a symmetric IPS refutation in general heavily depends on the choice of $\Gamma$. The bigger $\Gamma$ is and hence the more symmetric the circuits are required to be, the bigger we can expect the proof size to be.

We generally write \emph{$\Gamma$-sym-IPS} to mean the proof system allowing only $\Gamma$-symmetric proofs.  Of course, this only makes sense for sets $\F$ of polynomials that are $\Gamma$-invariant.  For these, as we show below, $\Gamma$-sym-IPS is a complete proof system.  Where $\Gamma$ is clear from context, we may write just sym-IPS.

Let $\Gamma$ and $\Gamma'$ be groups acting on $X$ with $\Gamma \leq \Gamma'$.  We say that $\Gamma'$-sym-IPS $p$-simulates $\Gamma$-sym-IPS if for every $\Gamma'$-invariant $\F$ and a $\Gamma$-sym-IPS refutation $C$ of $\F$, there is a $\Gamma'$-sym-IPS refutation of $\F$ with size polynomial in the size of $C$.  We also use this notion of $p$-simulation in the context of restricted proof systems, such as linear or bounded degree systems.  For instance, in Theorem~\ref{thm:deg_IPS=deg_sym_IPS} below, we show that (in suitably defined cases) $\degKIPS{k}  \leq_p \deg_k\text{-}\sIPS$.  This means that for each $k$, there is a fixed polynomial $p$ such that whenever a $\Gamma$-invariant $\F$ has any (i.e.\ $\{\id\}$-symmetric) IPS proof of size $n$ and degree $k$, it also has a $\Gamma$-symmetric IPS proof of size $p(n)$ and degree $k$.

\subsection{Completeness of  symmetric IPS}

A first natural question that arises when we restrict ourselves to symmetric proofs is whether every unsatisfiable $\Gamma$-symmetric polynomial equation system $\F$ has a $\Gamma$-symmetric refutation.  We show that this is indeed the case, in the sense that any IPS proof can be suitably symmetrised, though this may entail an exponential blowup in the size of the proof. On the other hand, there are $\F$ for which there is no symmetric \emph{linear} proof, which contrasts with the fact that linear IPS (without symmetry requirements) is complete.

We first prove the completeness of the symmetric IPS and then provide a simple counterexample for completeness for$\symIPSlin$. 
\begin{theorem}
	\label{thm:symIPScomplete}
	Let $\F$ be a $\Gamma$-symmetric system of polynomial equations.
	If $\F$ is unsatisfiable, then there
	is a $\Gamma$-symmetric $\IPS$ refutation of $\F$.
\end{theorem}
\begin{proof}
	Since IPS is complete, there is a certificate $C(\vec x, \vec y)$
	of unsatisfiability of $\F$, computed by some algebraic circuit $C$. We construct a $\Gamma$-symmetric circuit $\Circ^{sym}$ with the same semantics: For every $\pi \in \Gamma$, we introduce a copy of $\pi(C)$ in such a way that all these $\pi(C)$, for all $\pi \in C$, are identified at their input gates and otherwise disjoint. 
	To finish the construction, we add a multiplication gate $g_\times$ as the output of $\Circ^{sym}$. 
	It multiplies the outputs of all circuits $\pi(C)$, for all $\pi \in \Gamma$.
	The resulting circuit $\Circ^{sym}$ is $\Gamma$-symmetric by construction because every $\pi \in \Gamma$ extends to a circuit automorphism
	that maps each subcircuit $\pi'(C)$ to $(\pi \circ \pi')(C)$. We can also verify that $\Circ^{sym}$ is again a refutation:
	\begin{align}
		\Circ^{sym}(\vec x, \vec 0)
		 &= \prod_{\pi \in \Gamma}\pi (C(\vec x, \vec 0))
		 = \prod_{\pi \in \Gamma}\underbrace{C(\pi \vec x, \pi \vec 0)}_{= 0}
		= 0.\\
		{C}^{sym}(\vec x, \vec f)
		 &= \prod_{\pi \in \Gamma}\pi (C(\vec x, \vec f))
		 = \prod_{\pi \in \Gamma}\underbrace{C(\pi \vec x,\pi \vec f)}_{= 1}
		= 1.
	\end{align}	
	In the last equality of (1) and (2) we used that $0$ and $1$ are $\Gamma$-symmetric polynomials, so if $C(\vec x, \vec 0) = 0$, then also $C(\pi \vec x, \pi \vec 0) = 0$ for every $\pi \in \Gamma$ (and likewise for $1$).
	In the penultimate equality of (2), we also use the $\Gamma$-invariance of $\Ff$ and the fact that the action of $\Gamma$ on $Y$ is exactly as given by the lift of its action on $X$ to $\Ff$.
\end{proof}
This relatively naive construction blows up the size of the circuit by a factor of $|\Gamma|$, which may be as large as $|X|! \cdot |Y|!$. So even though there always exists a symmetric refutation, this may in general be much larger than the smallest asymmetric refutation.

The following example shows that there are cases where symmetric \emph{linear} refutations do not exist, regardless of the circuit size.
\begin{example}
	\label{ex:incompletenessOfSymHilbert}
	Let the variable set be $X = \{x_1, x_1^*, x_2, x_2^*\}$ and let $\Gamma \leq \Sym(X)$ be the group that is generated by $\{\pi,\pi^*\}$ defined as follows: $\pi = (x_1 \ x_2) \circ (x_1^* \ x_2^*)$, and $\pi^* = (x_1 \ x_1^*) \circ (x_2 \ x_2^*)$. That is, $\pi$ exchanges $1$ and $2$, and $\pi^*$ exchanges non-star with star. Consider the following equations over $\bbF_2$:\\
	\begin{tabular}{ccc}
		 (1) $x_1 + x_2 = 1$       &   (3) $x_1^*+x_2 = 1$       &  (5) $x_1 + x_1^* = 1$       \\
		  (2) $x_1^* + x_2^* = 1$ & (4) $x_1+x_2^* = 1$ & (6) $x_2 + x_2^* = 1$\\
	\end{tabular}	 
	~\\
	The equations are partitioned into three $\Gamma$-orbits. Equation (1) and (2) form an orbit, equation (3) and (4) as well, and equation (5) and (6).
	The system is unsatisfiable over $\bbF_2$ (and its algebraic closure) because $(1)+(4)+(6)$ is an IPS certificate of unsatisfiability.
	Nonetheless, there is no linear $\Gamma$-symmetric refutation $C(\vec x, \vec y)$: 
	Every monomial in such a certificate $C(\vec x, \vec y)$ would contain exactly one $\vec y$-variable. 
	In order to have $C(\vec x, \vec f) = 1$, there must be a degree-$1$ monomial in $C(\vec x, \vec y)$, i.e.\ consisting only of a single $\vec y$-variable $y_i$. But then, the entire orbit of $y_i$ must appear in $C(\vec x, \vec y)$ due to symmetry. As each orbit of equations has even size and the field has characteristic $2$, the constant terms in $C(\vec x, \vec f)$ sum up to $0$. Hence, there is no symmetric $\vec y$-linear polynomial with $C(\vec x, \vec f) = 1$.
\end{example}	 

On the other hand, we can show that as long we are working over a field $\Field$ of characteristic either $0$ or coprime with the order of $\Gamma$, then any $\Gamma$-invariant set $\F$ of polynomials has a $\Gamma$-symmetric linear refutation.  This follows from Corollary \ref{cor:completenessOfSymLinIPS} below.

\subsection{Symmetry in Bounded-Degree IPS}

Recall that for any constant $k \in \bbN$, $\deg_k$-IPS is the restriction of IPS where the polynomials computed at each gate have degree at most $k$. This bounded-degree fragment is of special interest because of its relation to the bounded-degree polynomial calculus, whose expressive power is related to important logics from finite model theory and also to the well-known Weisfeiler-Leman graph isomorphism algorithm (see Section \ref{sec:appplications}).
We show that constant-degree IPS proofs can be symmetrised efficiently, under certain assumptions on the field and the symmetry group (which are satisfied in most interesting cases). Thus, in the bounded-degree regime, requiring proofs to be symmetric is essentially no restriction.

\begin{theorem}
	\label{thm:deg_IPS=deg_sym_IPS}
	Let $\Field$ be a field and let $\Gamma$ be a group
	such that either $\mathrm{char}(\Field) = 0$, or $\bbF$ has positive characteristic and $|\Gamma|$ and $\mathrm{char}(\Field)$ are coprime. Let $k \in \bbN$ be a constant. 
	Then for every $\Gamma$-invariant polynomial equation system $\F \subseteq \Field[X]$ that possesses a $\deg_k$-$\IPS$ refutation $\Circ$, there also exists a \emph{$\Gamma$-symmetric} refutation $\Circ^{sym}$ with $|\Circ^{sym}| \leq \Oo(|\F|^{k}) \leq \Oo(|\Circ|^{k})$.
	If $\Circ$ is $\vec{y}$-linear, then so is $\Circ^{sym}$.
\end{theorem}
\begin{proof}
	Let $C(\vec x,\vec y)$ be a certificate of unsatisfiability of the polynomial equation system $\F$, computed by a circuit with semantic degree $k$. Let $M \subseteq \bbF[X \cup Y]$ be the set of monomials appearing in $C(\vec x,\vec y)$.
	Then the certificate $C(\vec x, \vec y)$ can be written as
	$
	\sum_{m \in M}c_m \cdot m(\vec x, \vec y),
	$
	where $c_m \in \Field$ and every $m \in M$ has degree at most $k$.
	We define
	$
	C^{sym}(\vec x, \vec y) := |\Gamma|^{-1} \cdot \sum_{\pi \in \Gamma}\pi C(\vec x, \vec y).
	$
	This polynomial $C^{sym}(\vec x, \vec y)$ is also an IPS certificate of unsatisfiability 
	of $\F$ since  
	\begin{enumerate}[label=(\arabic*)]
		\item $C^{sym}(\vec x, \vec 0) = 
		|\Gamma|^{-1} \cdot \sum_{\pi \in \Gamma} \pi C(\vec x, \vec 0) = 0$, and
		\item $C^{sym}(\vec x, \vec f) = 
		|\Gamma|^{-1} \cdot \sum_{\pi \in \Gamma} \pi C(\vec x, \vec f) = |\Gamma|^{-1} \cdot |\Gamma| \cdot 1 = 1$. 
	\end{enumerate}
	Note that $|\Gamma|^{-1}$ is defined in $\bbF$ since either the characteristic of $\bbF$ is zero, or it is coprime with $|\Gamma|$.  
	Let $\Gamma M$ be the closure of $M$ under the action of $\Gamma$, i.e.
	$
	\Gamma M := \bigcup_{\pi \in \Gamma} \pi(M).   
	$
	Now $C^{sym}(\vec x, \vec y)$ has the form 
	\begin{align*}
		C^{sym}(\vec x, \vec y) = 
		|\Gamma|^{-1} \sum_{\pi \in \Gamma}  \sum_{m \in M} c_m \cdot \pi(m(\vec x, \vec y)) 
		= |\Gamma|^{-1} \sum_{m \in \Gamma M} \Big( \sum_{\pi \in \Gamma} c_{\pi^{-1}(m)}  \Big) \cdot m
	\end{align*}
	The number of distinct monomials in $C^{sym}(\vec x, \vec y)$ is at most $\Oo((|X|+|Y|+1)^{k})$ because all monomials have degree at most $k$. 
	Thus, the circuit that just expresses $C^{sym}(\vec x, \vec y)$ as a sum of monomials has polynomial size, and it it also symmetric because $C^{sym}(\vec x, \vec y)$ is symmetric by construction.
	By our convention (see Section \ref{sec:IPS}), $|C| \geq |X|+|Y|$, so $|C^{sym}| \leq \poly(|C|)$. If $\Circ$ is $\vec{y}$-linear, then so is $\Circ^{sym}$ because if $M$ contains only $\vec y$-linear monomials, then this is still true for $\Gamma M$.
\end{proof}

\begin{corollary}
	\label{cor:boundedDegPC<=symIPS}
	Let $\Field$ be a field and let $\Gamma$ be a group
	such that either $\mathrm{char}(\Field) = 0$, or $\bbF$ has positive characteristic and $|\Gamma|$ and $\mathrm{char}(\Field)$ are coprime. Let $k \in \bbN$ be a constant. Then for this field and symmetry group,
	$
	\deg_k\text{-}\mathrm{PC} \leq_p \deg_k\text{-}\symIPSlin.
	$
\end{corollary}	
\begin{proof}
	Any degree-$k$ PC refutation can be translated in a straightforward way into a degree-$k$ $\IPSlin$-refutation (see \cite[Proposition 3.4]{IPS}). This can be efficiently $\Gamma$-symmetrised using Theorem~\ref{thm:deg_IPS=deg_sym_IPS}.
\end{proof}

\begin{corollary}
	\label{cor:completenessOfSymLinIPS}
	Let $\Field$ be a field and let $\Gamma$ be a group
	such that either $\charact(\Field) = 0$, or $\bbF$ has positive characteristic and $|\Gamma|$ and $\charact(\Field)$ are coprime.  Then, any $\Gamma$-invariant unsatisfiable set of polynomials $\F$ over this field has a $\Gamma$-symmetric linear refutation.
\end{corollary}	
\begin{proof}
	If $\Ff$ is an unsatisfiable polynomial equation system, then it has a PC refutation by the completeness of polynomial calculus. Let $k$ be its degree. By Corollary \ref{cor:boundedDegPC<=symIPS}, there exists a $\Gamma\text{-}\deg_k\text{-}\symIPSlin{}$-refutation.
\end{proof}

\section{Applications in Finite Model Theory and Graph Isomorphism}
\label{sec:appplications}
We now turn to other well-studied symmetry-invariant formalisms and show that they are subsumed by the symmetric IPS in a certain sense. These formalisms are logics from finite model theory, specifically \emph{fixed-point logic with counting} (FPC) and \emph{Choiceless Polynomial Time} (CPT). Drawing on previous work \cite{ggpp, Pago}, we show how their expressive power relates to the power of sym-IPS, specifically with respect to the graph isomorphism problem.

\subsection{Simulating fixed-point logic with counting in sym-IPS}
The evaluation of any sentence $\psi$ in fixed-point logic with counting in a given finite structure $\AA$ can be simulated by a bounded-degree symmetric IPS proof over $\bbQ$ in the following sense. For every fixed FPC-sentence $\psi$ and structure $\AA$, there is an axiom system $\Ff_\psi(\AA)$ that expresses the existence of a winning strategy in the model-checking game for $\AA \models \psi$. This is an instance of a \emph{threshold safety game} \cite{threshold_games_FPC}.
An IPS refutation of $\Ff_\psi(\AA)$ is then a witness for the fact that $\AA \models \psi$. 
The axiom system $\Ff_\psi(\AA)$ is FOC-interpretable in $\AA$, meaning that it is FO-definable in $\AA$ extended with a numeric sort (see Section \ref{sec:logics}). In total, the problem of deciding the existence of a bounded-degree sym-IPS refutation for a given axiom system is complete for FPC under efficient symmetry-preserving reductions:
\begin{theorem}
	\label{thm:simulateFPCinIPS}
	For every $\FPC$-sentence $\psi$ with signature $\tau$, there exists an $\mathrm{FOC}$-definable mapping $\Ff_\psi$ that takes every finite $\tau$-structure $\AA$ to a polynomial equation system $\Ff_\psi(\AA)$ over $\bbQ$ such that:
	\begin{enumerate}
		\item $|\Ff_\psi(\AA)| \leq \poly(|\AA|)$.
		\item $\Aut(\AA)$ has a natural action on the variables of $\Ff_\psi(\AA)$, and $\Ff_\psi(\AA)$ is $\Aut(\AA)$-invariant. 
		\item $\Ff_\psi(\AA)$ is unsatisfiable if, and only if, $\AA \models \psi$.
		\item If $\Ff_\psi(\AA)$ is unsatisfiable, then $\Ff_\psi(\AA)$ has an $\Aut(\AA)$-symmetric $\deg_2\text{-}\IPSlin$ refutation of size $\poly(|\Ff_\psi(\AA)|)$.
	\end{enumerate}	
\end{theorem}	
This is mainly a consequence of Theorem 4.4 in \cite{ggpp}. There, the desired mapping $\Ff_\psi$ is constructed and it is shown that $\Ff_\psi(\AA)$ has a degree-$2$ PC refutation $R$ over $\bbQ$ if and only if $\AA \models \psi$. 
By Corollary \ref{cor:boundedDegPC<=symIPS}, there also exists a $\Aut(\AA)$-symmetric $\deg_2\text{-}\IPSlin$-refutation of size $\poly(|R|)$, and even of size $\poly(|\Ff_\psi(\AA)|)$ (see Theorem~\ref{thm:deg_IPS=deg_sym_IPS}). This proves Item 4 from the theorem. The first two items follow from the properties of FOC-interpretations and the third item is due to the construction in \cite{ggpp}.

\subsection{Symmetric IPS proofs of graph non-isomorphism}
Now we pass on from FPC to a stronger logic, namely \emph{Choiceless Polynomial Time} (CPT) and show that this, too, can in a certain sense be simulated in$\symIPSlin{}$. The greater model-theoretic expressiveness of CPT is reflected on the proof system side in the fact that we (provably) need to go beyond the bounded-degree regime. 
Another difference to the simulation of FPC is that we now consider a fixed problem, namely
\emph{graph isomorphism}. We show that$\symIPSlin{}$ efficiently distinguishes all graphs (using their natural symmetries as the symmetry group) that are also distinguishable in CPT. 

\paragraph*{Distinguishing graphs in an algebraic proof system} 
Whenever we speak of \emph{graphs} in this section, they may be vertex and edge coloured. 
Given two graphs $G$ and $H$, there is a standard system of polynomial equations $\Piso(G,H)$ \cite{GI, Pago} whose solutions encode isomorphisms between $G$ and $H$. Thus, $G$ and $H$ are non-isomorphic if, and only if, $\Piso(G,H)$ is unsatisfiable.
The variable set of $\Piso(G,H)$ is $X \coloneqq \{x_{vw} \mid v \in V(G), w \in V(H), v \sim w\}$, where $\sim \subseteq V(G) \times V(H)$ is the relation that contains all $(v,w)$ such that $v$ and $w$ have the same colour.
The polynomials of $\Piso(G,H)$ are: $\sum_{\stackrel{v \in V(G)}{v \sim w}} x_{vw} - 1$ for each $w \in V(H)$, $\sum_{\stackrel{w \in V(H)}{v \sim w}} x_{vw} - 1 $ for each $v \in V(G)$, and $x_{vw}x_{v'w'}$ for all $v,v' \in V(G), w,w' \in V(H)$ with $v \sim w, v' \sim w'$ and such that $vv' \mapsto ww'$ is not a local isomorphism.
The Boolean axioms $x^2 - x$ for all $x \in X$ are also part of $\Piso(G,H)$. The idea behind this formulation is that a satisfying Boolean assignment to the variables in $X$ encodes an isomorphism from $G$ to $H$ in the sense that $v \in V(G)$ is mapped to $w \in V(H)$ if, and only if, $x_{vw}$ is assigned to $1$.
Note that $\Aut(G) \times \Aut(H)$ acts naturally on $X$: Let $(\pi_G, \pi_H) \in \Aut(G) \times \Aut(H)$. Then $(\pi_G, \pi_H)(x_{vw}) = x_{\pi_G(v)\pi_H(w)}$. It is not hard to see that $\Piso(G,H)$ is invariant under this group action: The polynomials associated with vertices in $G$ and $H$ are clearly symmetric and the polynomials $x_{vw}x_{v'w'}$ that forbid local non-isomorphisms depend on the edges and non-edges, which are preserved by $\Aut(G) \times \Aut(H)$.

When we say that an algebraic proof system distinguishes two graphs $G$ and $H$, we mean that it admits a refutation of $\Piso(G,H)$. 
If $\Kk$ is a class of graphs, then we say that$\symIPSlin$ \emph{efficiently distinguishes} all non-isomorphic graphs in $\Kk$ if there exists a polynomial $p(n)$ such that for any two non-isomorphic $G,H \in \Kk$, there exists an $\Aut(G) \times \Aut(H)$-symmetric $\IPSlin$-refutation $C$ of $\Piso(G,H)$ of size $|C| \leq p(|\Piso(G,H)|)$ over the field $\bbQ$.

\paragraph*{Distinguishing graphs in Choiceless Polynomial Time}
For any pair of non-isomorphic graphs $G$ and $H$, there is some formula (say of first-order logic) that distinguishes them.  For a class of graphs, we are interested in obtaining bounds (say on the number of variables or other parameters) of the minimum distinguishing formulas for pairs of non-isomorphic graphs from the class.  Here we are particularly interested in the number of variables of an $\FPC$ formula, or the resource bounds of a CPT sentence.  And, we relate these to bounds on the IPS refutation of $\Piso(G,H)$.

\begin{definition}[Distinguishing graphs in CPT, \cite{Pago}]
	\label{def:distinguishingInCPT}
	Let $\Kk$ be a class of graphs. We say that $\CPT$ distinguishes all graphs in $\Kk$ if there exists a polynomial $p(n)$ and a constant $k \in \bbN$ such that for any two non-isomorphic $G, H \in \Kk$, there exists a sentence $\Pi \in \CPT(p(n))$ with $\leq k$ variables such that $G \models \Pi$ and $H \not\models \Pi$. 
\end{definition}	
Recall from Section \ref{sec:logics} that $\CPT(p(n))$ is the fragment of $\CPT$ whose sentences have resource bound at most $p(n)$.
Similarly, the $k$-variable fragment of FPC distinguishes all graphs in $\Kk$ if there exists a distinguishing FPC-sentence with $\leq k$ variables for all $G \not\cong H$ in $\Kk$.
\begin{theorem}
	\label{thm:summaryGI}
	Let $\Kk$ be a graph class. 
	\vspace{-0.2cm}
	\begin{enumerate}
		\item If there is a $k \in \bbN$ such that the $k$-variable fragment of $\FPC$ distinguishes all non-isomorphic graphs in $\Kk$, then so does symmetric $\deg_k\text{-}\IPS$.
		\item If $\CPT$ distinguishes all non-isomorphic graphs in $\Kk$, then$\symIPSlin$ efficiently distinguishes them.
	\end{enumerate}
\end{theorem}	
Note that the situation in the first statement is equivalent to the non-isomorphic graphs in $\Kk$ being distinguishable by the well-known $k$-dimensional \emph{Weisfeiler Leman} algorithm \cite{sandraSiglog}.
The first part then follows from \cite[Theorem 4.4]{GI}, which states that $k$-Weisfeiler-Leman-distinguishable graphs can also be distinguished in the degree-$k$ PC. This can be $p$-simulated by the degree-$k$ sym-IPS by Corollary \ref{cor:boundedDegPC<=symIPS}.
To prove the second part, we use Theorem 1 from \cite{Pago}. It shows that if all non-isomorphic graphs in $\Kk$ are CPT-distinguishable in the sense of Definition \ref{def:distinguishingInCPT}, then they are also distinguishable in the degree-$3$ \emph{extended polynomial calculus} (EPC), and the refutations have polynomial size. As discussed in the conclusion of \cite{Pago}, this EPC refutation is symmetric in the sense that its extension axioms are closed under the action of $\Aut(G) \times \Aut(H)$.
To conclude the second item of Theorem \ref{thm:summaryGI}, we show that any symmetric bounded-degree EPC refutation can be simulated efficiently in$\symIPSlin{}$. The proof, including the definition of \emph{symmetric EPC}, is deferred to the appendix (Theorem \ref{thm:symdegEPC<=symIPS}).
Another result from \cite{Pago} immediately gives us the following separation between the bounded-degree (symmetric) IPS and its unbounded version.
\begin{theorem}
	\label{thm:separationOfSymBoundedDegree}
	There exists a sequence $(G_n,H_n)_{n \in \bbN}$ of pairs of non-isomorphic graphs such that $\Piso(G_n,H_n)$ has a polynomial-size $\Aut(G) \times \Aut(H)$-symmetric $\IPSlin$-refutation but there is no $k \in \bbN$ such that for all $n \in \bbN$, $\Piso(G_n,H_n)$ has a $\mathrm{deg}_k\text{-}\sIPS$-refutation.
\end{theorem}	
\begin{proof}
	Theorem 3 in \cite{Pago} yields exactly this statement for the symmetric degree-$3$ EPC and the degree-$k$ PC. 
	By Theorem \ref{thm:symdegEPC<=symIPS}, the symmetric degree-$3$ EPC can be $p$-simulated by$\symIPSlin{}$. Now suppose for a contradiction that $\Piso(G_n,H_n)$ had a degree-$k$ $\sIPS$-refutation for some constant $k$, for all $n \in \bbN$. This is in particular a degree-$k$ IPS refutation. 
	Using existing results, it can be shown that on instances of constant degree (which we have here), bounded-degree IPS refutations can be simulated in bounded-degree PC. A proof of this is given in Theorem \ref{thm:simulationsOfAlgebraicSystemsInIPSconstantDegree} in the appendix. This yields a contradiction to \cite[Theorem 3]{Pago}, which states that the bounded-degree PC cannot refute $\Piso(G_n,H_n)$ for all $n \in \bbN$.
\end{proof}	
The graphs $(G_n,H_n)$ that are used as hard instances for the bounded-degree IPS are the well-known \emph{Cai-Fürer-Immerman} (CFI) graphs \cite{GI, caifurimm92}, equipped with a linear order on their base graphs. These are standard examples of graphs that are indistinguishable in bounded-variable counting logic and hence FPC.
The above theorem tells us that these graphs are in fact efficiently distinguishable in$\symIPSlin$. This is because there is a CPT-sentence which can tell $G_n$ and $H_n$ apart, for all $n \in \bbN$ -- that sentence uses a sophisticated ``circuit-like'' construction due to \cite{dawar2008}. Via Theorem \ref{thm:summaryGI}, this translates into a polynomial-size$\symIPSlin{}$-refutation. In the next section, we study a different well-known formulation of the CFI graph isomorphism problem as a system of linear equations over a finite field. We show that also that presentation of the problem admits polynomial-size$\symIPSlin{}$-refutations.


\section{Upper bounds}

\subsection{The Cai-Fürer-Immerman equations}
\label{sec:CFI}

The algebraic formulation of the isomorphism problem of Cai-Fürer-Immerman graphs is the following system of equations over $\bbF_2$ (see e.g.\ \cite{atseriasBulatovDawar}).
\begin{restatable}[CFI equations]{definition}{CFIdefinition}
	\label{def:CFI}
	Let $G = (V,E)$ be a 3-regular undirected connected graph. Let $u \in V(G)$ be some fixed distinguished vertex and let $a \in \{0,1\}$. 
	The variable set is $X = \{x^e_i \mid e \in E, i \in \bbF_2\}$. The equations are
	\begin{align*}
		x_i^{e} + x_j^{f}  + x_k^{g} = i+j+k \mod 2  & & \text{      for every } v \in V \setminus \{u\}, \{e,f,g\} = E(v),\\
		&  & \text{ and every } i,j,k \in \bbF_2\\
		x_i^{e} + x_j^{f}  + x_k^{g} = i+j+k+a \mod 2 & & \text{ for vertex } u, \{e,f,g\} = E(u),\\
		&  & \text{ and every } i,j,k \in \bbF_2\\
		x_0^{e} + x_1^{e} = 1 & & \text{ for every } e \in E
	\end{align*}	
	This, together with the Boolean axioms for every variable, is the linear equation system $\CFI(G, u, a)$ over $\bbF_2$.
\end{restatable}
This system is satisfiable if, and only if, $a = 0$ \cite{atseriasBulatovDawar}. 
The typical symmetries that are associated with CFI graphs (over linearly ordered base graphs) are called ``edge flips''. With regards to the equation system, this means the following. Let $\Gamma$ be the subgroup of the Boolean vector space $(\bbF_2^{E}, \oplus)$ which consists only of those vectors $\pi$ such that $\sum_{e \in E(v)} \pi(e) = 0 \mod 2$ for every $v \in V$. 
The action of this group on $X$ is given by $\pi(x^e_i) = x^e_{i+\pi_e}$, where addition is in $\bbF_2$, and $\pi_e$ denotes the entry of $\pi$ at index $e \in E$. It is easy to check that $\CFI(G, u, a)$ is $\Gamma$-invariant. 
\begin{restatable}{theorem}{CFIupperbound}
	\label{thm:symmetricF2RefutationCFI}
	Let $(G_n)_{n \in \bbN}$ be an arbitrary family of $3$-regular graphs, and let $\Gamma_n$ be the subgroup of $\bbF_2^{E(G_n)}$ defined above.
	For every $n \in \bbN$, there exists a \emph{non-$\vec{y}$-linear} $\Gamma_n$-$\symIPS$-refutation over $\bbF_2$ of the unsatisfiable equation system $\CFI(G_n, u_n, 1)$ (regardless of the choice of $u_n \in V(G_n))$, which has size at most $\poly(|\CFI(G_n, u_n, 1)|)$. 
\end{restatable}	
In particular, the theorem is true if $(G_n)_{n \in \bbN}$ is a family of unbounded treewidth. Such families of graphs are the ones for which the equation systems $\CFI(G_n, u_n, 1)$ are not distinguishable from their satisfiable counterparts $\CFI(G_n, u_n, 0)$ in bounded-variable counting logic \cite{atseriasBulatovDawar}.
This also means that $\CFI(G_n, u_n, 1)$ has no bounded-degree PC refutation: The existence of such a refutation is definable in bounded-variable counting logic \cite{ggpp}, and hence, bounded-variable counting logic would be able to distinguish $\CFI(G_n, u_n, 1)$ from $\CFI(G_n, u_n, 0)$ if bounded-degree PC could refute $\CFI(G_n, u_n, 1)$. Thus, this theorem provides another example for the separation of the bounded-degree PC/IPS from the unbounded-degree version.

The construction of the refutation is quite involved, so we have to defer the proof of Theorem \ref{thm:symmetricF2RefutationCFI} to the appendix. We just remark that the circuit for the IPS certificate is deeply nested and computes polynomials of linear degree, so it is challenging to ensure that each gate only has a small number of automorphic images under the action of $\Gamma$. In fact, we have not been able to accomplish this with a $\vec{y}$-linear certificate, so the smallest possible$\symIPSlin{}$-refutation we know is exponential:
\begin{restatable}{theorem}{CFIupperboundLinear}
	\label{thm:symHilbertLikeCFI}
	Let the setting be as in Theorem \ref{thm:symmetricF2RefutationCFI}.
	For every $n \in \bbN$, there exists a $\Gamma_n$-$\symIPSlin$-refutation over $\bbF_2$ of the unsatisfiable equation system $\CFI(\Gg_n, u_n, 1)$, which has size at most $\Oo(2^{|E_n|})$.
\end{restatable}	
It is an interesting open question whether the CFI equations provide an exponential separation between$\symIPS{}$ and$\symIPSlin{}$, that is, whether the upper bound in the above theorem can be improved or not. As discussed in the previous section, the graph isomorphism formulation of the CFI equations does admit a small$\symIPSlin{}$-refutation but it may well be that this is not the case for $\CFI(G_n, u_n, 1)$.

\subsection{Subset sum}

\begin{definition}[Subset sum, \cite{lowerbounds_IPS}]
	Let $n \in \bbN$, let $\bbF$ be a field with $\text{char}(\bbF) > n$, and let $\beta \in \bbF \setminus \{0,...,n\}$. The \emph{subset sum} instance $\Ff_{\text{sum}(\vec x)}(n,\bbF,\beta)$ has variable set
	$X = \{x_1,...,x_n\}$, and the axiom
	$
	\sum_{i=1}^n x_i - \beta = 0,
	$
	along with the Boolean axioms $x_i^2 - x_i = 0$ for all $i \in [n]$.
	The \emph{lifted subset sum} instance $\Ff_{\text{sum}(\vec x\vec y)}(n,\bbF,\beta)$ has variable set $X \cup \{y_1,...,y_n\}$ and the axiom 	$
	\sum_{i=1}^n x_iy_i- \beta = 0,
	$
	along with the Boolean axioms for all variables in $X \cup Y$.
\end{definition}	
It is clear that $\Ff_{\text{sum}(\vec x)}(n,\bbF,\beta)$ and $\Ff_{\text{sum}(\vec x \vec y)}(n,\bbF,\beta)$ are unsatisfiable for any choice of $n, \bbF, \beta$ as required in the definition. Moreover, $\Ff_{\text{sum}(\vec x)}(n,\bbF,\beta)$ is $\Sym_n$-symmetric with respect to the obvious action on $X$, and $\Ff_{\text{sum}(\vec x\vec y)}(n,\bbF,\beta)$ is $\Sym_n$-symmetric with respect to the simultaneous action on $X \cup Y$.
It is proven in \cite[Proposition 5.3]{lowerbounds_IPS} that $\Ff_{\text{sum}(\vec x)}(n,\bbF,\beta)$ has no $\deg_k\text{-}\IPSlin$ refutation for any $k < n$.
Moreover, \cite{lowerbounds_IPS} shows $\Ff_{\text{sum}(\vec x)}(n,\bbF,\beta)$ to be hard for \emph{sparse IPS} (where circuits are just allowed to be sums of monomials) and $\Ff_{\text{sum}(\vec x \vec y)}(n,\bbF,\beta)$ to be hard for roABPs with a fixed variable order, and so-called depth-3 powering formulas.


Subset sum and its liftings are thus a natural starting point in the quest for sym-IPS lower bounds, especially because the variable set is ``maximally symmetric'' and so it might be expected that the size of any symmetric refutation must be large. However, it turns out that at least for the two subset sum variants we study here, polynomial-size symmetric refutations do in fact exist. 
There are more complex liftings of the subset axiom that have been used for lower bounds against stronger fragments of IPS such as bounded product-depth circuits
\cite{govindasamy2022simple, hakoniemi2024functional}, and these may be promising candidates for sym-IPS lower bounds, too.

\begin{restatable}{theorem}{subsetsumUpperBound}
	\label{thm:symRefutationSubsetSum}
	The polynomial equation system $\Ff_{\text{sum}(\vec x)}(n,\bbF,\beta)$ has a $\Sym_n\text{-}\symIPSlin{}$-refutation of size at most $\poly(| \Ff_{\text{sum}}(n,\bbF,\beta) |)$, for all $n, \bbF, \beta$ such that the system is unsatisfiable.
	The same is true for $\Ff_{\text{sum}(\vec x \vec y)}(n,\bbF,\beta)$.
	 Moreover, there is no constant $k \in \bbN$ such that $\deg_k\text{-}\IPS$ can refute $\Ff_{\text{sum}(\vec x)}(n,\bbF,\beta)$ and $\Ff_{\text{sum}(\vec x \vec y)}(n,\bbF,\beta)$ for all $n \in \bbN$.
\end{restatable}	
The full proof of the theorem can be found in the appendix. In short, we show that the refutation given in Proposition B.1 in the appendix of \cite{lowerbounds_IPS} can be realised with polynomial-size symmetric circuits. The key ingredient are the elementary symmetric polynomials $S_{n,k} = \sum_{\stackrel{S \subseteq [n]}{|S| = k}} \prod_{i \in S} x_i$. They admit efficient symmetric circuits by \cite{shpilka2001depth}.

\subsection{Pigeonhole principle}

\begin{definition}
	For $n,m \in \bbN$, the $n$-to-$m$ \emph{pigeonhole principle} is the polynomial equation system $\PHP(n,m)$. Its variable set is $X = \{x_{ij} \mid i \in [n], j \in [m]\}$ and its equations are the Boolean axioms together with:
	\begin{align*}
		\sum_{j \in [m]} x_{ij} - 1 = 0&\text{    for every } i \in [n]\\
		x_{ij}x_{i'j} = 0& \text{    for every } j \in [m], i \neq i' \in [n]
	\end{align*}
\end{definition}	
Any $\{0,1\}$-valued solution to this system gives an injective function from $[n]$ to $[m]$ that maps $i$ to $j$ if $x_{ij} = 1$.  The equation $\sum_{j \in [m]} x_{ij} - 1 = 0$ guarantees that each value $i$ is mapped to exactly one $j$ and $x_{ij}x_{i'j}  = 0$ ensures that distinct $i$ and $i'$ are not mapped to the same value.
Whenever $n > m$, $\PHP(n,m)$ is thus unsatisfiable. This is the case over $\bbQ$, but also over every finite field.
It is easy to see that $\PHP(n,m)$ is invariant under $\Sym_n \times \Sym_m$, where the group action on $X$ is: $(\pi,\sigma)(x_{ij}) = x_{\pi(i)\sigma(j)}$. In this section, we focus on the pigeonhole principle $\PHP(n+1,n)$. It can be checked that over finite fields, $\PHP(n+1,n)$ does not admit a symmetric $\vec{y}$-linear refutation for all $n$ (analogous to Example \ref{ex:incompletenessOfSymHilbert}). Therefore, we only consider $\PHP(n+1,n)$ over $\bbQ$.

With no symmetry restriction in place, it is -- not surprisingly -- possible to refute $\PHP(n+1,n)$ efficiently in the IPS.  Indeed, the IPS $p$-simulates any Frege proof system~\cite{IPS}, and the pigeonhole principle, formulated in propositional logic has a polynomial-size Frege proof \cite{buss1987polynomial}. The proof constructed in \cite{buss1987polynomial}, however, proceeds along a linear order on $[n]$. Thus, a naive symmetrisation of it would require size $\Oo(n!)$.
We show that we can do much better than that, although we do not obtain a symmetric refutation of subexponential size. It is plausible that this is impossible, and we leave the precise complexity of symmetrically refuting the pigeonhole principle as an intriguing open problem.
\begin{restatable}{theorem}{PHPupperBound}
	\label{thm:upperboundSymPHP}
	There is a $(\Sym_{n+1} \times \Sym_n)\text{-}\symIPSlin$ refutation of $\PHP(n+1,n)$ of size $\Oo(3^n \cdot n)$ over the field $\bbQ$.
\end{restatable}	
The key part of the refutation is to compute, for every subset $D \subseteq [n+1]$ the sum over all monomials that encode injections from the pigeons in $D$ to the holes. Formally, let
$
B_D(\vec x) \coloneqq \sum_{\gamma \colon D \hookrightarrow [n]} \prod_{i \in D} x_{i\gamma(i)}.
$
The polymnomials $B_D(\vec x)$ can also be viewed as the sums of certain \emph{permanents}. The permanent of an $n \times n$-matrix is the polynomial $\sum_{\pi \in \Sym_n} \prod_{i \in [n]} x_{i\pi(i)}$. The polynomial $B_D(\vec x)$ is the sum over the permanents of all $|D| \times |D|$-submatrices of $D \times [n]$. 
This hints at the potential hardness of symmetric refutations for the PHP because we know that the permanent admits no symmetric circuit representation of subexponential size \cite{DawarW20}. 
In the proof of Theorem \ref{thm:upperboundSymPHP} in the appendix, we construct an $\Oo(3^n)$-size symmetric circuit for computing the $B_D$ for all $D \subseteq [n+1]$ and show how this yields a refutation for $\PHP(n+1,n)$.

\section{Conclusion and future work}

This article initiates the study of how symmetry in IPS proofs affects their complexity. We identify the following promising directions for future research:
Firstly, we would like to obtain matching lower bounds for the exponential upper bounds we have established -- that is, for \emph{linear} symmetric IPS refutations of the CFI equations (which would show an exponential gap between linear and non-linear refutations), and for the pigeonhole principle. 
Secondly, the question in how far the functional lower bound method \cite{lowerbounds_IPS, govindasamy2022simple, hakoniemi2024functional} can be combined with our framework deserves a deeper investigation. We have shown that the subset sum axiom and one of its liftings do admit small symmetric proofs but this does not rule out a lower bound via more complex subset sum variants such as in \cite{govindasamy2022simple, hakoniemi2024functional}. 
Also, it is worth studying if the symmetry restriction we consider here, combined with the functional lower bound method from \cite{govindasamy2022simple, hakoniemi2024functional} can extend the scope of that technique. So far, the functional lower bound method has only been applied to instances with a \emph{single axiom}, and it provably fails on encodings of Boolean formulas \cite{hakoniemi2024functional}; can these limitations of the functional lower bound method be overcome by restricting to symmetric refutations?
Finally, we ask if the connection between IPS lower bounds for \emph{Boolean CNFs} and the separation of $\VP$ and $\VNP$ established by Grochow and Pitassi \cite{IPS} also holds in a symmetric sense. 
That is, for a suitably defined symmetric analogue of $\VNP$, is it true that every unsatisfiable CNF has an IPS refutation in symmetric $\VNP$?

\bibliographystyle{plainurl}
\bibliography{ref.bib}

\newpage
\appendix

\section{Simulation of the extended polynomial calculus in symmetric IPS}
The simulation presented here is essential in the proof of Theorem \ref{thm:summaryGI}.
The \textit{extended polynomial calculus} (EPC) is an extension of the polynomial calculus with
the \textit{extension rule}:
\begin{center}
	{\Large$\frac{}{~  x_e - f_e ~}$},
	for any $f_e\in \Field[X \uplus X_E]$
\end{center}
where $X$ is the variable set of the instance and $x_e$ is an
\textit{extension variable} distinct from the variables in $X$, and $X_E$ denotes the set of all extension variables. For an EPC refutation to be correct, there must exist a linear order $\preceq$ on the set of extension variables in it, which satisfies: In every extension axiom $x_e - f_e$, the polynomial $f_e$ only contains extension variables that are smaller than $x_e$ with respect to $\preceq$. We refer to this as a \emph{hierarchical ordering} of the extension axioms. An EPC refutation of $\F$ involving a set $\Ee$ of extension axioms can also be viewed as a (non-extended) PC refutation of $\F \cup \Ee$. By $\deg_k$-EPC we denote the restriction of EPC where every proof line has degree at most $k$. Admissible instances for these proofs systems must also have degree at most $k$.

For the proof of Theorem \ref{thm:summaryGI}, it is necessary to simulate \emph{symmetric} $\deg_3$-EPC refutations by symmetric refutations in $\IPSlin$. The $\deg_3$-EPC refutations that we wish to simulate are symmetric in the following relatively weak sense.

\paragraph*{Symmetric EPC}
Let $\Gamma$ be a group acting on $X$.
Let $\mathcal E = \{e_1,e_2,...,e_r\}$ be a set of hierarchical extension axioms over $X$ with
$e_1\preceq e_2\preceq ... \preceq e_r$
and $e_i = x_{e_i} - f_{e_i}$, where the $x_{e_i} \notin X$ are extension variables. 
We say that the action of $\Gamma$ \emph{naturally lifts} to $\Ee$ if the following three conditions hold: 
\begin{enumerate}
	\item There exists a partition of $\Ee$ into classes $C_1,C_2,...,C_m$ such that for each $i \in [m]$, for every extension axiom $e = x_e - f_e \in C_i$, $f_e$ only contains extension variables from classes $C_j$ with $j < i$.
	\item There are no two extension axioms $x_{e} - f_e$ and $x_{e'} - f_{e'}$ in $\Ee$ such that $f_e = f_{e'}$ and $x_e \neq x_{e'}$.
	\item For every $C_i$ and every $\pi \in \Gamma$, $\pi(C_i) = C_i$.
\end{enumerate}
For the third item, the action of $\pi \in \Gamma$ on $\Ee$ is defined inductively as follows. For $e = (x_e - f_e) \in C_1$, $\pi(e)$ is the unique extension axiom $(x_{e'} - f_{e'}) \in \Ee$ such that $f_{e'} = \pi(f_e)$ (uniqueness is ensured by the second condition). Here, $\pi(f_e)$ is defined via the action of $\Gamma$ on $X$. For $e \in C_{i+1}$, $\pi(e)$ is defined in the same way, where we use the fact that $\pi(f_e)$ is defined by the induction hypothesis, as $f_e \in C_i$.

\begin{definition}[sym-EPC]
	\label{def:symEPC}
	Let $\Gamma$ be a group, $X$ be a $\Gamma$-set, and let $\Field$ be a field. \\
	Let $\F \subseteq \Field[X]$ be a $\Gamma$-invariant unsatisfiable system of equations.\\
	A $\Gamma$-symmetric $\EPC$ refutation $R$ of $\F$ is an $\EPC$ refutation of $\F$ such that the action of $\Gamma$ on $X$ \emph{naturally lifts} to the set $\mathcal E$ of extension axioms appearing in $R$.
\end{definition}
Note that this definition does not require all extension axioms in $\Ee$ to actually play a role in the refutation but they of course contribute to its size (i.e.\ proof length).

The graph isomorphism refutations that we need to simulate for Theorem \ref{thm:summaryGI} come from \cite[Theorem 1]{Pago}, and they are in $\deg_3$-sym-EPC for $\Gamma=\Aut(G) \times \Aut(H)$, where $G$ and $H$ are the two graphs to be distinguished.
The following theorem shows that these refutations can be simulated in$\symIPSlin$, which allows us to conclude the second statement of Theorem \ref{thm:summaryGI}. 
\begin{theorem} 
	\label{thm:symdegEPC<=symIPS}
	Let $\Field$ be a field and let $\Gamma$ be a group
	such that either $\charact(\Field) = 0$, or $\bbF$ has positive characteristic and $|\Gamma|$ and $\charact(\Field)$ are coprime. Let $k \in \N$. Then
	\begin{center}
		$\mathrm{deg}_k \text{-}\mathrm{sym}\text{-}\EPC  \leq_p  \symIPSlin$.
	\end{center}
\end{theorem}
\begin{proof}
	Let $X$ be a $\Gamma$-set and let $\F \subseteq \Field[X]$
	be a $\Gamma$-invariant unsatisfiable polynomial system of equations.
	Let $R$ be a $\mathrm{deg}_k \text{-}\mathrm{sym}\text{-}\EPC$ refutation of $\F$ with 	extension axioms $\mathcal{E}$. 
	Let $e_1 \preceq e_2 \preceq ...\preceq e_r$ 
	and $x_{e_1}\preceq x_{e_2} \preceq ...\preceq x_{e_r}$ be
	the hierarchical ordering of the extension axioms and variables.
	We think of $R$ as a PC refutation of $\F \cup \Ee$, and let $C_R$ be the $\IPSlin$ refutation that $p$-simulates this PC refutation (which exists by \cite[Proposition 3.4]{IPS}).
	
	It computes a polynomial 
	$C_R(\vec x, \vec x_e, \vec y, \vec y_e)\in \Field[\vec x, \vec x_e,\vec y, \vec y_e ]$ which is an IPS certificate of unsatisfiability of $\Ff \cup \Ee$. The semantic degree of $C_R$ is at most $k$ because each gate of $C_R$ computes a polynomial $p(\vec x, \vec x_e, \vec y, \vec y_e)$ such that $p(\vec x, \vec x_e, \vec f, \vec e)$ is a proof line of $R$. Hence, $\deg(p(\vec x, \vec x_e, \vec f, \vec e)) \leq k$, and $\deg(p(\vec x, \vec x_e, \vec y, \vec y_e))$ can only be smaller.
	Note also that $\Ff \cup \Ee$ is a $\Gamma$-invariant axiom system because the action of $\Gamma$ on $X$ naturally lifts to $\Ee$.
	Therefore, we can apply Theorem~\ref{thm:deg_IPS=deg_sym_IPS} to $C_R$ and obtain a $\Gamma$-symmetric$\symIPSlin$ refutation $C^{sym}_R$ of $\Ff \cup \Ee$. 
	All steps so far increased the proof size at most polynomially, so we have $|C^{sym}_R| \leq \poly(|R|)$.\\
	The circuit $C^{sym}_R$ computes a polynomial in $\Field[\vec x, \vec x_e,\vec y, \vec y_e]$, so we still need to get rid of the $\vec x_e$- and $\vec y_e$-variables. We first define $C'^{sym}_R$ as the circuit obtained from $C^{sym}_R$ by replacing the $\vec y_e$-variables with the corresponding extension axioms $x_{e_i} - f_{e_i}$. 
	Next, we replace in order $i=1, ..., r$ every extension-variable input gate $x_{e_i}$ by a subcircuit computing $f_{e_i}(\vec x, \vec f_e) \in \bbF[X]$. Call the resulting circuit $C^{sym}$. This circuit is still $\vec y$-linear. To prove that $|C^{sym}| \leq \poly(|C^{sym}_R|) \leq \poly(|R|)$, we just need to see that each $f_{e_i}$ is computable with a subcircuit of polynomial size. But this is clearly satisfied because we started with a degree-$k$ EPC refutation, so each $f_{e_i}$ also has degree $\leq k$. Hence, it can be represented as a sum of only polynomially many monomials. 
	Now we show that $C^{sym}$ is indeed a refutation of $\vec f$.
	We have
	\[
	C^{sym}(\vec x, \vec f) = C_R(\vec x, \vec f_e, \vec f, \vec e) = \theta C_R(\vec x, \vec x_e, \vec f, \vec e) = \theta 1 = 1,
	\]
	where $\theta \colon \bbF[X \cup X_E] \to \bbF[X]$ denotes the operator which replaces in a given polynomial $p$ each extension variable $x_{e_i}$ with the polynomial $f_{e_i}(\vec x, \vec f_e)$. In the third equality, we use the fact that $\theta$ is a ring homomorphism, so $\theta$ commutes with the operations of the algebraic circuit. Secondly, 
	\[
	C^{sym}(\vec x, \vec 0) =  C_R(\vec x, \vec f_e, \vec 0, \vec e) = C_R(\vec x, \theta \vec x_e, \vec 0, \theta \vec e) = \theta C_R(\vec x, \vec x_e, \vec 0, \vec 0)  = \theta 0 = 0.
	\]
	Here, we made use of the fact that $\theta e_i = \theta(x_{e_i} - f_{e_i}) = 0$ for every $i \in [r]$, and again that $\theta$ commutes with the operations of the circuit.
	
	The final property we have to check is that $C^{sym}$ is $\Gamma$-symmetric. The circuit $C^{sym}_R$ is $\Gamma$-symmetric by construction, with respect to the action of $\Gamma$ on $X \cup X_E \cup Y \cup Y_E$, where $X_E$ denotes the set of extension variables and $Y_E$ the set of input variables for the extension axioms. In $C^{sym}$, the input gates with labels in $X_E \cup Y_E$ are replaced with subcircuits computing the corresponding $f_{e_i}$ and $e_i$. Because $\Ee$ is $\Gamma$-invariant, these subcircuits can be chosen so that they are symmetric to each other. That is, for each $\pi \in \Gamma$, and each $f_{e_i}$, the subcircuit computing $f_{e_i}$ is mapped by $\pi$ to the subcircuit computing $\pi f_{e_i}$, and likewise for the extension axioms $e_i$. 
\end{proof}

\section{Equivalence of bounded-degree IPS and PC}
We provide details for the claim that on instances of constant degree, bounded-degree IPS refutations can be simulated in bounded-degree PC. This is used in the proof of Theorem \ref{thm:separationOfSymBoundedDegree}.
\begin{theorem}
	\label{thm:simulationsOfAlgebraicSystemsInIPSconstantDegree}
	On instances of \emph{constant degree}, we have the following relationships between the bounded-degree fragments of the proof systems.\footnote{The relationship between the degree bounds in PC and $\IPSlin$ is not the same in each direction of the mutual simulation. This is because the degree of the axioms is counted in a PC refutation while in IPS, each axiom just contributes degree $1$, as it is replaced with a $\vec y$-variable.}    
	\[
	\text{bounded-degree }\mathrm{PC} \equiv_p \text{bounded-degree }\IPSlin \equiv_p \text{bounded-degree } \IPS.
	\]
\end{theorem}
First, consider the equivalence $\text{bounded-degree }\mathrm{PC} \equiv_p \text{bounded-degree }\IPSlin$. 
This is partially proved by \cite[Proposition 3.4]{IPS}, which establishes $p$-equivalence between PC and the so-called ``determinantal'' or ``skew'' $\IPSlin$. This is the restriction of $\IPSlin$ to \emph{skew} circuits, in which each multiplication gate has fan-in $\leq 2$ and at least one of its inputs is an input gate of the circuit.
It can be easily seen that the $p$-simulation from \cite[Proposition 3.4]{IPS} is degree-preserving, up to the additive cost of the axiom degree in going from an $\IPSlin$- to a PC-refutation (replacing $\vec y$-variables with the axioms they stand for). This shows: $\text{bounded-degree }\mathrm{PC} \equiv_p \text{bounded-degree skew  }\IPSlin$. Moreover, it trivially holds: $\text{bounded-degree skew }\IPSlin \leq_p \text{bounded-degree } \IPS$. 
Thus, it remains to prove that on instances of constant degree, we also have: 
\[
\text{bounded-degree } \IPS \leq_p \text{bounded-degree skew }\IPSlin.
\]
This is dealt with in the following lemma.
\begin{lemma}
	\label{lem:constantDegreeAllowsSkewisation}
	Let $\F$ be a polynomial system of equations with $\deg(\F) \leq d$.\\ 
	If there exists a $\deg_k$-$\IPS$ refutation $C$ of $\F$, then there is also a skew $\vec y$-linear $\deg_{dk}$-$\IPS$ refutation of size at most $\BigO(|\Circ|^{dk})$.
\end{lemma}
\begin{proof}
	Let $C$ be a $\deg_k$-$\IPS$ refutation of $\F \subseteq \bbF[X]$.
	The circuit $C$ computes the polynomial
	$C(\vec x,\vec y) = \sum_{i\in I} c_im_i(\vec x,\vec y)$,
	for distinct monomials $m_i \in \Field[X \cup Y]$. Since the degree of
	$C$ is bounded by $k$, we have $|I|= \BigO((|X| + |Y|+1)^{k})$, and we can assume $C$ to be a circuit that just represents this polynomial as a sum of monomials, so $|C| = \BigO((|X| + |Y|+1)^{k})$.
	To make the circuit $\vec y$-linear, we replace in each monomial $m_i(\vec x, \vec y)$ with more than one $\vec y$-variable all but one of them with subcircuits for the respective axioms. That is, the variable $y_j$ is replaced with a circuit computing the axiom $f_j \in \F$. Call the resulting $\vec y$-linear
	circuit $C_{lin}$. Since $\deg(\F) \leq d$, the degree of each subcircuit computing an axiom $f_j$ is at most $d$. Thus, the total degree of $C_{lin}$ is at most $dk$ because $C(\vec x,\vec y) = \sum_{i\in I} c_im_i(\vec x,\vec y)$ has degree $\leq k$, and in $C_{lin}$, each variable is replaced by a polynomial in $\bbF[X]$ of degree $\leq d$. Furthermore, it is clear by construction that $C_{lin}(\vec x, \vec f) = C(\vec x, \vec f) = 1$. Also, $C_{lin}(\vec x, \vec 0) = 0$ because each of the monomials in $C(\vec x, \vec y)$ and hence also in $C_{lin}(\vec x, \vec y)$ contains at least one $\vec y$-variable (else, $C(\vec x, \vec 0) \neq 0$, and $C$ would not be an IPS refutation). Thus, $C_{lin}$ is also an IPS refutation of $\Ff$. It remains to transform this into a skew circuit. We have
	\[
	C_{lin}(\vec x,\vec y) = \sum_{i\in I} c_i \cdot m_i(\vec x,y_{j_i},\vec f),
	\]
	where $m_i(\vec x,y_{j_i},\vec f)$ denotes the monomial $m_i$ with all but one $\vec y$-variable $y_{j_i}$ replaced with the corresponding axiom. Each of the polynomials $m_i(\vec x,y_{j_i},\vec f)$ has the form $y_{j_i} \cdot q_i(\vec x)$ for some $q_i \in \bbF[X]$. This expression can be computed by a skew circuit as follows: For each monomial $x_{1} \cdot ...  \cdot x_t$ in $q_i$, compute $y_{j_i} \cdot x_{1} \cdot ...  \cdot x_t$ by a sequence of $t$ many degree-$2$ multiplication gates. Then take the sum over all these monomials to obtain $y_{j_i} \cdot q_i(\vec x)$. Finally, take the sum over all these $c_i \cdot y_{j_i} \cdot q_i(\vec x)$, for all $i \in I$. This describes a skew $\vec y$-linear degree-$dk$ circuit computing $C_{lin}(\vec x,\vec y)$. Its size is at most $\Oo((|X| + |Y|+1)^{dk})$ because every $m_i((\vec x,y_{j_i},\vec f)$ has degree at most $dk$, so this is the maximum number of monomials we have to compute.
	By our definition of IPS proof size (see Section \ref{sec:IPS}), $|C| \geq |X| + |Y|$, so this is in $\Oo(|C|^{dk})$
\end{proof}

\section{Upper bound results }

\subsection{Cai-Fürer-Immerman equations}
We recall the set-up from Section \ref{sec:CFI}.
\CFIdefinition*
The $\vec y$-variables associated with these equations are:
\begin{align*}
	Y := \{ y^v[e \mapsto i,f \mapsto j,g \mapsto k] &\mid v \in V, i,j,k \in \bbF_2, \{e,f,g\} = E(v)   \} \\
	& \uplus \{ y^e \mid e \in E \} \\
	& \uplus \{ y_{\text{Bool}}[e,i] \mid e \in E(v), i \in \bbF_2  \}  .
\end{align*}
The symmetry group we consider is the subgroup $\Gamma$ of the Boolean vector space $(\bbF_2^{E}, \oplus)$ which consists only of those vectors $\pi$ such that $\sum_{e \in E(v)} \pi(e) = 0 \mod 2$ for every $v \in V$. 
The action of this group on $X$ is given by $\pi(x^e_i) = x^e_{i+\pi_e}$, where addition is in $\bbF_2$, and $\pi_e$ denotes the entry of $\pi$ at index $e \in E$. 
The corresponding action on $Y$ is:
\[
\pi(y^v[e \mapsto i,f \mapsto j,g \mapsto k]) = y^v[e \mapsto i+\pi_{e},f \mapsto j+\pi_{f},g \mapsto k+\pi_{g}],
\]
$\pi(y^e) = y^e$, and $\pi(y_{\text{Bool}}[e,i]) = y_{\text{Bool}}[e,i+\pi_e]$. We now prove:

\CFIupperbound*

In the following, we say that a monomial contains a \emph{conflict} if for some $e \in E$, it contains both $x^e_0$ and $x^e_1$. It follows from the edge equations that such a conflict-monomial evaluates to zero. This is also derivable with a small symmetric circuit:
\begin{lemma}
	\label{lem:symmetricCircuitForE01}
	For every edge $e \in E$, there is a constant-size $\Gamma$-symmetric circuit $C_{e01}$ such that $C_{e01}(\vec x, \CFI(G, u, 1)) = x^e_0x^e_1$ and $C_{e01}(\vec x, \vec 0) = 0$.
\end{lemma}
\begin{proof}
	We can write
	\[
	C_{e01}(\vec x, \vec y) = y^e \cdot x^e_0x^e_1 + y_{\text{Bool}}[e,0] \cdot x^e_1 + y_{\text{Bool}}[e,1] \cdot x^e_0.
	\]
	Let $C_{e01}$ be the circuit that computes the polynomial exactly in this way.
	It is easy to see that this is symmetric under flipping $e$ (i.e.\ under the unit vector in $\Gamma$ with a $1$ at index $e$) because $y_e$ is fixed by this edge flip. Flipping other edges has no effect on this circuit. When the $y$-variables are substituted with the axioms, this becomes in $\bbF_2$:
	\[
	C_{e01}(\vec x,  \CFI(\Gg, u, 1)) = (x^e_0 + x^e_1 +1) \cdot x^e_0x^e_1  + ((x^e_0)^2 + x^e_0) \cdot x^e_1 + ((x^e_1)^2 + x^e_1) \cdot x^e_0 =  x^e_0x^e_1\\.
	\]
\end{proof}

In the following, when we consider a vertex $v \in V$, we always call its incident edges $e$, $f$ and $g$. We introduce the shorthand notation $\llbracket P? \rrbracket$, for a Boolean predicate $P$. This evaluates to $1$ if $P$ holds, and $0$ otherwise.
For every $v \in V$, define:
\[
\tau(v) := \sum_{\stackrel{(i,j,k) \in \bbF_2^3}{i+j+k = \llbracket v \neq u? \rrbracket \mod 2}} x_i^{e} x^{f}_j x^{g}_k + 1
\]
\[
\tautilde(v) := \sum_{\stackrel{(i,j,k) \in \bbF_2^3}{i+j+k = \llbracket v = u? \rrbracket \mod 2}} x_i^{e} x^{f}_j x^{g}_k
\]
\begin{lemma}
	\label{lem:tauPolynomialsComputable}
	For each $v \in V$, there exist polynomial-size circuits $C_v$ and $\widetilde{C}_v$ such that 
	\begin{enumerate}
		\item $C_v(\vec x, \CFI(G_n, u_n, 1)) = \tau(v)$ and $\widetilde{C}_v(\vec x, \CFI(G_n, u_n, 1)) = \tautilde(v)$. 
		\item $C_v(\vec x, \vec 0) = 0$ and $\widetilde{C}_v(\vec x, \vec 0) = 0$.
		\item $C_v$ and $\widetilde{C}_v$ are $\Gamma$-symmetric.
	\end{enumerate}	
\end{lemma}	
\begin{proof}
	We define the following polynomials, each of which is obviously computable by a $\Gamma$-symmetric circuit with one multiplication gate.
	\begin{align*}
		C'_v(\vec x, \vec y) &:= \prod_{\stackrel{(i,j,k) \in \bbF_2^3}{i+j+k = \llbracket v \neq u? \rrbracket \mod 2}} y^v[e \mapsto i, f \mapsto j, g \mapsto k].\\
		\widetilde{C}'_v(\vec x, \vec y) &:= \prod_{\stackrel{(i,j,k) \in \bbF_2^3}{i+j+k = \llbracket v \neq u? \rrbracket \mod 2}} y^v[e \mapsto i, f \mapsto j, g \mapsto k].
	\end{align*}
	Both these circuits evaluate to zero if the $\vec y$-variables are substituted with zero. When we substitute the CFI axioms for the $\vec y$-variables, we get the following expression (taking into account that any monomial with an even number of occurrences is cancelled):
	\begin{align*}
		C'_v(\vec x, \CFI(G, u, 1)) = 1 &+ \sum_{i+j+k = \llbracket v \neq u? \rrbracket} (x^e_i)^2x^f_jx^g_k + x^e_i(x^f_j)^2x^g_k + x^e_ix^f_j(x^g_k)^2\\
		&+ \sum_{\stackrel{\{ a,b\} \subseteq E(v), i,j \in \bbF_2}{a \neq b}} x^a_ix^b_j + (x^a_i)^2x^b_j + x^a_i(x^b_j)^2  + \sum_{a\in E(v), i \in \bbF_2} (x^a_i)^2\\
		&+x^e_0x^e_1 \cdot \Big(x^e_0x^e_1+\sum_{i,j \in \bbF_2} (x^f_i)^2+(x^g_j)^2 + (x_i^f+ x_i^g) \cdot (x_0^e + x_1^e)  \Big)\\
		&+ x^f_0x^f_1 \cdot \Big(x^f_0x^f_1+\sum_{i,j \in \bbF_2} (x^e_i)^2+(x^g_j)^2 + (x_i^e+ x_i^g) \cdot (x_0^f + x_1^f) \Big)\\
		&+ x^g_0x^g_1 \cdot \Big(x^g_0x^g_1+\sum_{i,j \in \bbF_2} (x^e_i)^2+(x^f_j)^2 + (x_i^e+ x_i^f) \cdot (x_0^g + x_1^g) \Big).   
	\end{align*}	
	To obtain $\tau(v)$ from this expression, we have to do the following. First, replace each monomial of the form $(x^e_i)^2x^f_jx^g_k$ with its multilinearised version. This can be done by multiplying the respective Boolean axiom for the squared variable with the other variables in the monomial and adding the result to the sum. It can be seen that doing this in the obvious way does not break $\Gamma$-symmetries. Next, we want to cancel all monomials in the second line of the above expression. To do this, we again multilinearise everything using the Boolean axioms. Then, for every $2$-element subset $\{a,b\} \subseteq E(v)$, we add the product $y^a \cdot y^b$ to the sum. When these $y$-variables are replaced with the edge axioms $(x^a_0 + x^a_1 +1) \cdot (x^b_0 + x^b_1 +1)$, then this cancels all degree-$2$ monomials in line $2$ and adds a $+1$ to the sum. The degree-$1$ monomials that are added this way are cancelled automatically because each degree-$1$ monomial is added twice. Thus, we still have to cancel $\sum_{a\in E(v), i \in \bbF_2} x^a_i$ (after it has been multilinearised). This can be done by adding all the three edge axioms $y^a$ for $a \in E(v)$, that is, $(x^a_0 + x^a_1 +1)$. This introduces another $+1$, so it cancels the one that was introduced when eliminating the degree-$2$ monomials. All of this is $\Gamma$-symmetric.\\
	Finally, to cancel the last three lines, we use the symmetric circuits $C_{e01}, C_{f01}, C_{g01}$ from Lemma \ref{lem:symmetricCircuitForE01}, and multiply them with the correct polynomial that appears in the expression above. This is also $\Gamma$-symmetric. The sum of $C'_v$ and all the described auxiliary circuits is then $C_v$. In particular, all the added auxiliary circuits evaluate to zero when the $y$-variables are set to zero because all $\vec x$-polynomials are always multiplied with $y$-variables in them.\\
	
	For the other circuit, we get:
	\begin{align*}
		\widetilde{C}'_v(\vec x, \CFI(G, u, 1)) =  &\sum_{i+j+k = \llbracket v = u? \rrbracket} (x^e_i)^2x^f_jx^g_k + x^e_i(x^f_j)^2x^g_k + x^e_ix^f_j(x^g_k)^2\\
		&+x^e_0x^e_1 \cdot \Big(x^e_0x^e_1+\sum_{i,j \in \bbF_2} (x^f_i)^2+(x^g_j)^2 +  (x_i^f+ x_i^g) \cdot (x_0^e + x_1^e) \Big)\\
		&+x^f_0x^f_1\cdot \Big(x^f_0x^f_1+\sum_{i,j \in \bbF_2} (x^e_i)^2+(x^g_j)^2 +  (x_i^e+ x_i^g) \cdot (x_0^f + x_1^f) \Big)\\
		&+ x^g_0x^g_1 \cdot \Big(x^g_0x^g_1+\sum_{i,j \in \bbF_2} (x^e_i)^2+(x^f_j)^2 +  (x_i^e+ x_i^f) \cdot (x_0^g + x_1^g) \Big). 
	\end{align*}
	We can proceed similarly as with $C_v$ to obtain $\tautilde(v)$ from this by adding $\Gamma$-symmetric circuits to it.
\end{proof}	
Note that the circuits $C_v$ and $\widetilde{C}_v$ are clearly not $\vec y$-linear.

The CFI equation system is related to the so-called \emph{Tseitin equations}. In the Tseitin equation system defined by $(G,u,1)$, we have only one equation for each $v \in V$, instead of eight. Let $E(v) = \{e,f,g\}$. We identify the variables of this equation system with the edges. The Tseitin equation for every $v \neq u$ reads
\[
e+f+g = 0,
\] 
and for the special vertex $u$ we have the equation:
\[
e+f+g = 1,
\] 
Now fix any linear ordering $<$ on $V$ and let $v_1,v_2,...,v_n$ be the enumeration of $V$ according to $<$. Let $\mu_1 := \tau(v_1)$ and $\mutilde_1 := \tautilde(v_1)$.\\ 

For $i \in [n]$, let $W_i := \{v_1,...,v_i\}$, and $E_i := \bigcup_{w \in W_i} E(w)$. Let $\text{Sol}_1(W_i)$ be the set of all assignments $\lambda \colon E_i \to \bbF_2$ which are a correct solution for an odd number of Tseitin equations for vertices in $W_i$. Similarly, $\text{Sol}_0(W_i)$ is the set of all $\lambda \colon E_i \to \bbF_2$ that are correct solutions for an even number of Tseitin equations for vertices in $W_i$.\\
For all $1 < i \leq n$, define
\[
\mu_i \coloneqq  \sum_{\lambda \in \text{Sol}_0(W_i)} \prod_{e \in E_i} (x_{\lambda(e)}^{e} )^{\kappa(e)} + 1.
\]
\[
\mutilde_i \coloneqq \sum_{\lambda \in \text{Sol}_1(W_i)} \prod_{e \in E_i} (x_{\lambda(e)}^{e})^{\kappa(e)}.
\]
The exponent $\kappa(e)$ denotes how many endpoints of $e$ are in $W_i$, so this is either one or two.
\begin{lemma}
	\label{lem:computingMu}
	For each $i \in [n]$, there exist polynomial-size circuits $C_i$ and $\widetilde{C}_i$ such that 
	\begin{enumerate}
		\item $C_i(\vec x, \CFI(G_n, u_n, 1)) = \mu_i$ and $\widetilde{C}_i(\vec x, \CFI(G_n, u_n, 1)) = \mutilde_i$. 
		\item $C_i(\vec x, \vec 0) = 0$ and $\widetilde{C}_i(\vec x, \vec 0) = 0$.
		\item $C_i$ and $\widetilde{C}_i$ are $\Gamma$-symmetric.
	\end{enumerate}	
\end{lemma}

Once we have this, we are done. The following proves Theorem \ref{thm:symmetricF2RefutationCFI}.
\begin{corollary}
	\label{cor:symRefutationCFI}
	There is a polynomial-size $\Gamma$-symmetric $\IPS$-refutation for $\CFI(G, u, 1)$.
\end{corollary}
\begin{proof}
	The circuit for $\mu_n$ is the desired refutation.
	We have $\mu_n = 1$ because $\text{Sol}_0(W_n) = \text{Sol}_0(V)$ is empty. This can be seen as follows: Since $G$ is $3$-regular, $|V|$ must be even. Thus, for each assignment $\lambda\colon E \to \bbF_2$ in $\text{Sol}_0(V)$, there exists a partition $V = V_0 \uplus V_1$ with $|V_0| \equiv |V_1| \equiv 0 \mod 2$ such that $\lambda$ correctly solves the Tseitin equations for every vertex in $V_0$ and is incorrect for the Tseitin equation for each $v \in V_1$. 
	Such a $\lambda$ cannot exist: Suppose $u \in V_0$. The rhs of the Tseitin equations for $V_0$ sum up to $1$ then. The lhs of these equations then also sum up to $1$ under $\lambda$. Semantically, the sum over the lhs is the sum over all $e \in E$ in the cut between $V_0$ and $V_1$. Thus, also the lhs and the rhs of the equations for vertices in $V_1$ must sum to $1$. However, $\lambda$ does not satisfy any equation for vertices in $V_1$, so they all evaluate to $1$ (because $u \notin V_1$). Since $|V_1|$ is even, the sum over these equations is thus $0$ and not $1$. We can argue symmetrically if $u \in V_1$. 
\end{proof}	

Now let us prove Lemma \ref{lem:computingMu}.
\begin{proof}
	We construct the circuits $C_i$ and 
	$\widetilde{C}_i$ by induction on $i$.
	The case $i = 1$ is handled by Lemma \ref{lem:tauPolynomialsComputable}.
	To compute $\mu_i$ and $\mutilde_i$ for $i > 1$, we define the following circuits:
	\begin{align*}
		C_i &\coloneqq \mu_{i-1} \cdot \tau(v_i) + \mu_{i-1} + \tau(v_i)  + \mutilde_{i-1} \cdot \tautilde(v_i).\\
		\widetilde{C}_i &\coloneqq \mu_{i-1} \cdot \tautilde(v_i) +  \tautilde(v_i) + \mutilde_{i-1} \cdot \tau(v_i) + \mutilde_{i-1}.
	\end{align*}	
	For $\mu_{i-1}, \mutilde_{i-1}, \tau(v), \tautilde(v)$ we use the circuits we have by induction and by Lemma \ref{lem:tauPolynomialsComputable}.
	It directly follows from the induction hypothesis and from Lemma \ref{lem:tauPolynomialsComputable} that $C_i$ and $\widetilde{C}_i$ are both $\Gamma$-symmetric. Likewise, item 2 follows from these. It remains to check that $\mu_{i-1} \cdot \tau(v_i) + \mu_{i-1} + \tau(v_i) + \mutilde_{i-1} \cdot \tautilde(v_i) = \mu_i$ (and symmetrically for $\mutilde_i$). First of all, note that the summands $\mu_{i-1} + \tau(v_i)$ in $C_i$ are just there to cancel out the summands in $\mu_{i-1} \cdot \tau(v_i)$ that arise because both $\mu_{i-1}$ and $\tau(v_i)$ contain the summand $1$.\\
	By definition of the $\tau$- and $\mu$-polynomials, the monomials in $\mu_{i-1} \cdot \tau(v_i) + \mu_{i-1} + \tau(v_i)$ which do not contain a conflict are in bijection with the set of assignments $\lambda\colon E_i \to \bbF_2$ which do not satisfy the Tseitin equation for $v_i$ and satisfy an even number of Tseitin equations for vertices in $W_{i-1}$. Similarly, the conflict-free monomials in $\mutilde_{i-1} \cdot \tautilde(v_i)$ are those which code an assignment satisfying the equation for $v_i$ and an odd number of vertex equations in $W_{i-1}$. In total, these monomials code all assignments that satisfy an even number of Tseitin  equations for vertices in $W_i$. Note that each of these monomials appears exactly once in the sum, so it is not cancelled. Also, the exponent $\kappa(e)$ for any $e$ in the cut between $W_{i-1}$ and $\{v_i\}$ is $2$, as required. Each monomial in $C_i(\vec x, \CFI(G, u, 1))$ that contains a conflict appears an even number of times and thus cancels itself: Let $m$ be such a monomial. Then let $x_0^{e_1}x_1^{e_1}x_0^{e_2}x_1^{e_2}...x_0^{e_k}x_1^{e_k}$ for edges $\{e_1,...,e_k\}$ in the cut between $v_i$ and $W_{i-1}$ be the submonomial of $m$ containing exactly its conflicts.
	We show that $m$ appears exactly $2^k \equiv 0 \mod 2$ times in $C_i(\vec x, \CFI(G, u, 1))$: Let $F \subseteq \{e_1,...,e_k\}$ be any subset of the conflicting edges in $m$. Let $E^*_i := (E_i \setminus \{e_1,..., e_k \} )\uplus \{  (e_1,v), (e_1,W)..., (e_k, v), (e_k, W) \}$. That is, we ``double'' the conflicting edges and introduce one variable for each of its endpoints. Then let $\lambda_F \colon E^*_i \to \bbF_2$ be the assignment which is as encoded by $m$ on the non-conflicting edges, and for each conflicting edge $e_i$, $\lambda_F(e_i,v) := 1$ iff $e_i \in F$, and $\lambda_F(e_i,W) := 1-\lambda_F(e_i,v)$. We define a \emph{modified} version of the Tseitin equations for the vertices $W_{i-1} \cup v_i$ over variables $E_i^*$: In the equation associated with $v_i$, every variable $e_j \in \{e_1,...,e_k\}$ is replaced by $(e_j,v)$. In every other equation, every occurrence of a variable $e_j \in \{e_1,...,e_k\}$ is replaced by $(e_j,W)$. The other variables remain the same.\\
	\\
	\textbf{Claim:} For each $F \subseteq \{e_1,...,e_k\}$, $\lambda_F$ is either a correct solution for the modified $v_i$-Tseitin-equation and for an odd number of modified $W_{i-1}$-equations, or it is incorrect for the modified $v_i$-equation and correct for an even number of modified $W_{i-1}$-equations.\\
	\textit{Proof of claim.} For at least one $F$, this must be true because otherwise, $m$ would not appear in $C_i(\vec x, \CFI(\Gg, u, 1))$: The fact that $m$ appears in $C_i(\vec x, \CFI(\Gg, u, 1))$ means that for every $e_j \in \{e_1,...,e_k\}$, one of $\{x_0^{e_j}, x_1^{e_j}\}$ appears in $\mu_{i-1}$ and the other in $\tau(v_i)$, or one appears in $\mutilde_{i-1}$ and the other in $\tautilde(v_i)$. By the semantics of $\mu$ and $\tau$, this is exactly the statement of the claim for this particular $F$. To show that the claim is true for every other $F' \subseteq \{e_1,...,e_k\}$ as well, we argue that whenever the claim holds for some $\lambda_F$, it also holds for $\lambda_{F \triangle \{e_j\}}$, for any $j \in [k]$: Swapping the values of $(e_j, v)$ and $(e_j, W)$ means that the satisfaction status of the modified $v_i$-equation is changed, and also the satisfaction status of the modified equation for one vertex in $W_{i-1}$ is changed. Hence, the claim also holds for $\lambda_{F \triangle \{e_j\}}$, and by repeating this, it holds for every $F \subseteq \{e_1,...,e_k\}$. \qedsymbol \\
	\\
	Because of the claim, for each $F \subseteq  \{e_1,...,e_k\}$, $\lambda_F$ describes a distinct occurrence of the monomial $m$ (that is, a distinct way to compute it) in $\mu_{i-1} \cdot \tau(v_i) + \mutilde_{i-1} \cdot \tautilde(v_i)$ (again by the semantics of $\mu_{i-1}$ and $\tau(v_i)$). Thus, $m$ appears indeed exactly $2^k$ times and is therefore cancelled.
	In total, we have proved Lemma \ref{lem:computingMu}.
\end{proof}
The construction of the circuits $C_i$ and $\widetilde{C}_i$ in the above proof is based on the same idea as the hereditarily finite sets called $\mu$ and $\tau$ in the CPT-algorithm from \cite{dawar2008}. Both settings are quite different, but in both cases, $\mu_i$ symbolises that a certain aggregated parity in the subinstance on $W_i$ is even, and $\mutilde_i$ stands for the odd parity.

\paragraph*{Symmetric linear refutation}

Next, we construct a $\vec y$-linear symmetric refutation for the CFI equations, thus proving

\CFIupperboundLinear*
For the proof, we assume as before that we have
fixed some ordering $<$ on $V$. This extends to an ordering on the edges as well, and we can enumerate them $E = e_1,e_2,...,e_m$. We define the following $\Gamma$-symmetric polynomial
\begin{align*}
	B(\vec x, \vec y) := y^{e_1} &+ y^{e_2} \cdot (x^{e_1}_0+x^{e_1}_1) + y^{e_3} \cdot (x^{e_1}_0x^{e_2}_0 + x^{e_1}_1x^{e_2}_0 + x^{e_1}_1x^{e_2}_1 + x^{e_1}_0x^{e_2}_1) + ... \\
	&+ y^{e_m} \cdot \Big( \sum_{(i_1,...,i_{m-1}) \in \bbF_2^{m-1}} x^{e_1}_{i_1} \cdot ... \cdot x^{e_{m-1}}_{i_{m-1}} \Big)
\end{align*}
\begin{lemma}
	\label{lem:Bcomputable}
	The polynomial $B(\vec x, \vec y)$ is $\vec y$-linear and satisfies: 
	\begin{itemize}
		\item $B(\vec x, \vec 0) = 0$.
		\item $B(\vec x, \CFI(G,u,1)) = 1 + \sum_{(i_1,...,i_{m}) \in \bbF_2^{m}} x^{e_1}_{i_1} \cdot ... \cdot x^{e_{m}}_{i_{m}}$. 
		\item There exists a $\Gamma$-symmetric circuit $B$ with $|B| \leq \poly(| \CFI(G,u,1)|)$ that computes $B(\vec x, \vec y)$.
	\end{itemize}	
\end{lemma}	
\begin{proof}
	The fact that $B$ is $\vec y$-linear and the first property are easy to see. For the second property, note that $B(\vec x, \CFI(G,u,1))$ is a telescoping sum. The $i$-th summand introduces all monomials encoding Boolean assignments to the edges $\{e_1,...,e_i\}$, and it also cancels all assignments to the edges $\{e_1,...,e_{i-1}\}$, which remain from the $(i-1)$st summand. The ends of this telescope sum are $1$ and the sum over all assignments to all edges. To compute $B(\vec x, \vec y)$ with a polynomial-size circuit, we have to compute the sum $\sum_{(i_1,...,i_{j}) \in \bbF_2^{j}} x^{e_1}_{i_1} \cdot ... \cdot x^{e_{j}}_{i_{j}}$, for each $j \in [m-1]$, with a polynomial-size circuit. This can be done by writing it as the product $\prod_{i=1}^j (x^{e_i}_0 + x^{e_i}_1)$.    
\end{proof}		
Next, we show how to cancel all monomials except $1$ in $B(\vec x, \CFI(G,u,1))$. Let $A := \{ \lambda : E \to \bbF_2\}$ denote the set of all total assignments of Boolean values to edges. The group $\Gamma$ acts on $A$ in the sense that $\pi(\lambda)(e) = \lambda(e)+\pi_e$ for every $\pi \in \Gamma, e \in E$.
If we partition $A$ into its $\Gamma$-orbits, we see that each $\Gamma$-orbit $\Omega \subseteq A$ is characterised by the set $W \subseteq V$ of vertices at which the assignments in $\Omega$ are correct solutions for the Tseitin equations:
\begin{lemma}
	For every $\Gamma$-orbit $\Omega \subseteq A$, there exists a $W \subseteq V$ such that
	\[
	\Omega = \{ \lambda \in A \mid \forall v \in V: \lambda \text{ solves the Tseitin equation for } v \text{ iff } v \in W  \}.
	\]
\end{lemma}
\begin{proof}
	Every $\pi \in \Gamma$ flips an even number of incident edges at every $v \in V$ and hence, $\pi$ cannot change the set of vertices for which $\lambda \in A$ is a correct Tseitin solution. Conversely, if two $\lambda, \lambda' \in A$ 
	are correct Tseitin solutions for exactly the same vertices $W \subseteq V$, then there is a $\pi \in \Gamma$ with $\pi(\lambda) = \lambda'$. Indeed, let $F \subseteq E$ be the set of edges at which $\lambda$ and $\lambda'$ differ. Every vertex in $V$ is incident to an even number of edges in $F$; otherwise, there would be a vertex whose Tseitin equation is solved by $\lambda$ but not by $\lambda'$. Thus, the Boolean vector $\pi_F$, which flips exactly the edges in $F$, is in $\Gamma$, and $\pi_F(\lambda) = \lambda'$.
\end{proof}

Thus, we denote each orbit of $A$ as $\Omega_W$, for the appropriate $W \subseteq V$ that characterises the orbit in the sense of the above lemma.
\begin{lemma}
	\label{lem:noWisEmpty}
	There exists no $\Gamma$-orbit $\Omega_\emptyset \subseteq A$.
\end{lemma}	
\begin{proof}
	The orbit $\Omega_\emptyset$ contains all assignments $\lambda \in A$ which do not solve any Tseitin equation correctly. But such an assignment does not exist, for the same reason why the Tseitin equations are unsatisfiable. Since $G$ is $3$-regular and hence $|V|$ is even, the sum over the inverses of the right hand sides of all Tseitin equations is $1$ modulo $2$. Thus, for an assignment $\lambda$ that satisfies none of the Tseitin equations, we would have \[
	\sum_{\stackrel{v \in V}{\{e,f,g\} = E(v)}} \lambda(e) + \lambda(f) + \lambda(g) = 1 \mod 2
	\] But this is impossible because every edge in $E$ appears exactly twice in this sum, so it must be even.
\end{proof}	

We group the orbits $\Omega_W \subseteq A$ together, using the total order $v_1,v_2,...,v_n$ on the vertices. Let $O_1 := \{ \Omega_W \mid v_1 \in W  \}$. For $i > 1$, let $O_i := \{ \Omega_{W} \mid v_i \in W \text{ and } v_j \notin W \text{ for all } j <i   \}$.\\
Due to Lemma \ref{lem:noWisEmpty}, each orbit appears in one $O_i$, and by the definition of the $O_i$, it appears in exactly one of them.
The set $A$ and hence also the set of non-constant monomials in the sum $B(\vec x, \CFI(G,u,1))$ is thus partitioned into the $O_i$ (some of which may be empty but this does not matter). For simplicity, we also write $O_i$ for $\bigcup O_i$, i.e.\ treat each $O_i$ as a set of assignments rather than a set of sets of assignments.
If $\lambda\colon E' \to \bbF_2$ is some Boolean assignment to a subset of the edges, we write $m(\lambda) := \prod_{e \in E'} x^{e}_{\lambda(e)}$ for the monomial that naturally encodes this.
For every $O_i$, and every triple $(j,k,\ell) \in \bbF_2^3$ with $j+k+\ell = \llbracket v_i= u? \rrbracket \mod 2$, we define the polynomial
\[
A^{(j,k,\ell)}_i(\vec x) \coloneqq x^{e}_jx^f_kx^g_\ell \cdot \Big( \sum_{\stackrel{\lambda \colon E \setminus \{e,f,g\} \to \bbF_2}{\lambda \cup (e \mapsto j, f \mapsto k, g \mapsto \ell) \in O_i}} m(\lambda) \Big).       
\] 
Here, $\{e,f,g\} = E(v_i)$. Also, we define a version of this where the edges $\{e,f,g\}$ are missing:
\[
A'^{(j,k,\ell)}_i(\vec x) \coloneqq  \sum_{\stackrel{\lambda \colon E \setminus \{e,f,g\} \to \bbF_2}{\lambda \cup (e \mapsto j, f \mapsto k, g \mapsto \ell) \in O_i}} m(\lambda).       
\] 
In the following, we mean by $\Stab(e,f,g)$ the subgroup of $\Gamma$ consisting of all vectors which do not flip any of the edges $e,f,g$.
\begin{lemma}
	\label{lem:computingA}
	For every $i \in [n]$ and every $(j,k,\ell) \in \bbF_2^3$ with $j+k+\ell = \llbracket v_i= u? \rrbracket \mod 2$, the polynomials $A^{(j,k,\ell)}_i(\vec x)$ and $A'^{(j,k,\ell)}_i(\vec x)$ are computable by $\Stab(e,f,g)$-symmetric circuits of size $\Oo(|O_i|)$. Moreover, these circuits are such that every $\pi \in \Gamma \setminus \Stab(e,f,g)$ maps the circuits for $A^{(j,k,\ell)}_i(\vec x)$ and $A'^{(j,k,\ell)}_i(\vec x)$, respectively, to the circuits for $A^{(j+\pi_e,k+\pi_f,\ell+\pi_g)}_i(\vec x)$ and $A'^{(j+\pi_e,k+\pi_f,\ell+\pi_g)}_i(\vec x)$.
\end{lemma}	
\begin{proof}
	The set of assignments $L := \{\lambda \colon E \setminus \{e,f,g\} \to \bbF_2 \mid \lambda \cup (e \mapsto j, f \mapsto k, g \mapsto \ell) \in O_i  \}$ is the projection to $E \setminus \{e,f,g\}$ of all assignments in $O_i$ with certain prescribed values at $e,f,g$.
	Since $O_i$ is a union of $\Gamma$-orbits, every $\pi \in \Stab(e,f,g)$ stabilises $L$. Thus, the natural depth-2-circuit that computes the sum over all assignments in $L$ as a sum of products is $\Stab(e,f,g)$-symmetric. It requires the computation of at most $|O_i|$ monomials, so this is an upper bound on its size.
	For $A^{(j,k,\ell)}_i(\vec x)$, we have to multiply the result by $x^{e}_jx^f_kx^g_\ell$, which still results in a $\Stab(e,f,g)$-symmetric circuit. The additional statement that  every $\pi \in \Gamma \setminus \Stab(e,f,g)$ maps the circuits for $A^{(j,k,\ell)}_i(\vec x)$ and $A'^{(j,k,\ell)}_i(\vec x)$ to the circuits for $A^{(j+\pi_e,k+\pi_f,\ell+\pi_g)}_i(\vec x)$ and $A'^{(j+\pi_e,k+\pi_f,\ell+\pi_g)}_i(\vec x)$, is true because this holds for the polynomials  $A^{(j,k,\ell)}_i(\vec x)$ and $A'^{(j,k,\ell)}_i(\vec x)$ themselves. 
\end{proof}	

The above lemma is the only place where we incur the exponential cost of the refutation. 
We do not know if there are more efficient ways to compute $A^{(j,k,\ell)}_i(\vec x)$ and $A'^{(j,k,\ell)}_i(\vec x)$. 
The difficulty is that the assignments in $O_i$ all have the property that they do not satisfy the Tseitin equation for any $v_j$ with $j < i$. That is, we do not want all assignments here, and not even all assignments that satisfy an even number of Tseitin equations. The polynomial size non-$\vec y$-linear circuit we constructed in the proof of Lemma \ref{lem:computingMu} combined local assignments to the incident edges of each vertex to obtain the sum over exponentially many assignments to all vertices. In that case, incompatible local assignments cancelled because they appeared an even number of times. Here, we would have to take the product over the local assignments to $E(v_j)$ that do not satisfy the Tseitin equation for $v_j$, for every $j < i$. But we do not see a way to efficiently remove the monomials involving conflicts that would arise in this way (in particular, every monomial involving a conflict would be computed only once and would not cancel directly).\\

Now we multiply the polynomials $A^{(j,k,\ell)}_i(\vec x)$ and $A'^{(j,k,\ell)}_i(\vec x)$ in a suitable way with $y$-variables in order to obtain an IPS certificate. We define for every $i \in [n]$ the following polynomial.
\begin{align*}
	C_i(\vec x, \vec y) \coloneqq \sum_{\stackrel{(j,k,\ell) \in \bbF_2^3}{j+k+\ell = \llbracket v_i = u? \rrbracket}} &\Big(A^{(j,k,\ell)}_i(\vec x) \cdot y^{v_i}[e \mapsto j, f \mapsto k, g \mapsto \ell]\\
	&+ A'^{(j,k,\ell)}_i(\vec x) \cdot (y_{\text{Bool}}[e,j] \cdot x^f_kx^g_\ell \\
	&+ y_{\text{Bool}}[f,k]\cdot x^e_jx^g_\ell+ y_{\text{Bool}}[g,\ell]\cdot x^e_jx^f_k ) \Big).
\end{align*}
\begin{lemma}
	\label{lem:Ccomputable}
	The polynomial $C_i(\vec x, \vec y)$ is $\vec y$-linear and satisfies:
	\begin{itemize}
		\item $C_i(\vec x, \vec 0) = 0$.
		\item $C_i(\vec x, \CFI(G,u,1)) = \sum_{\lambda \in O_i} m(\lambda)$. 
		\item $C_i$ can be computed by a $\Gamma$-symmetric circuit of size $\Oo(|O_i|)$.
	\end{itemize}
\end{lemma}	
\begin{proof}
	The first property is easy to see. For the second property, note that the sum $\sum_{\stackrel{(j,k,\ell) \in \bbF_2^3}{j+k+\ell = \llbracket v_i = u? \rrbracket}} A^{(j,k,\ell)}_i$ contains $m(\lambda)$ for every $\lambda \in O_i$: Every $\lambda \in O_i$ is a correct solution for the Tseitin equation at vertex $v_i$, and $(j,k,\ell)$ ranges over all correct solutions for the $v_i$-equation.
	We see that $A^{(j,k,\ell)}_i \cdot y^{v_i}[e \mapsto j, f \mapsto k, g \mapsto \ell]$ gives us $m(\lambda)$, for every $\lambda \in O_i$ that sends $(e,f,g)$ to $(j,k,\ell)$ exactly thrice: Once with $x^e_j$ squared, once with $x^f_k$ squared, and once with $x^g_\ell$ squared. The respective Boolean axioms are used to multilinearise all these monomials. Since each monomial appears thrice, we get each monomial exactly once after addition in $\bbF_2$.\\
	The $\Gamma$-symmetric circuit for $C_i(\vec x, \vec y)$ looks exactly like the definition of $C_i(\vec x, \vec y)$, where we use the respective $\Stab(e,f,g)$-symmetric circuits from Lemma \ref{lem:computingA} to compute $A^{(j,k,\ell)}_i(\vec x)$ and $A'^{(j,k,\ell)}_i(\vec x)$. In total, this is a $\Gamma$-symmetric circuit because any $\pi \in \Gamma \setminus \Stab(e,f,g)$ just permutes the summands of the sum $\sum_{\stackrel{(j,k,\ell) \in \bbF_2^3}{j+k+\ell = \llbracket v_i = u? \rrbracket}} ...$, and the circuits for computing each summand are indeed mapped to each other by the action of $\pi$ (as mentioned in Lemma \ref{lem:computingA}). The size of the circuit for $C_i$ is a constant factor times the sizes of the circuits for $A^{(j,k,\ell)}_i(\vec x)$ and $A'^{(j,k,\ell)}_i(\vec x)$, which are $\Oo(|O_i|)$ by Lemma \ref{lem:computingA}. 
\end{proof}	

Putting all this together, we get
\[
1 = B(\vec x, \CFI(G,u,1)) + \sum_{i=1}^n C_i(\vec x,  \CFI(G,u,1)).
\]	
This is a $\vec y$-linear refutation for $\CFI(G,u,1)$, and by Lemmas \ref{lem:Bcomputable} and \ref{lem:Ccomputable}, it is computable by a $\Gamma$-symmetric circuit. Its size is dominated by the $C_i$. Their sizes sum up to $\Oo\Big(\sum_i |O_i| \Big) = \Oo(2^{|E|})$, as the $O_i$ partition the set of all $2^{|E|}$ many assignments $\lambda \colon E \to \bbF_2$.

\subsection{Subset sum}
\subsetsumUpperBound*
\begin{proof}
	We first deal with the case of $\Ff_{\text{sum}(\vec x)}(n,\bbF,\beta)$ and explain in the end, how the proof lifts to $\Ff_{\text{sum}(\vec x \vec y)}(n,\bbF,\beta)$. 
	In Proposition B.1 in the appendix of \cite{lowerbounds_IPS}, the key part of an IPS refutation
	of $\Ff_{\text{sum}(\vec x)}(n,\bbF,\beta)$ is explicitly constructed. The authors show that the axiom 	$\sum_{i=1}^n x_i - \beta$ has to be multiplied with the polynomial
	\[
	f(\vec x) = - \sum_{k=0}^n \frac{k!}{\prod_{j=0}^k (\beta -j)} S_{n,k},
	\]
	where \[
	S_{n,k} = \sum_{\stackrel{S \subseteq [n]}{|S| = k}} \prod_{i \in S} x_i
	\] is the $k$-th elementary symmetric polynomial. Concretely, \cite[Proposition B.1]{lowerbounds_IPS} states that $f(\vec x) \cdot \Big(\sum_{i=1}^n x_i - \beta \Big)$ is equal to $1$ on all Boolean assignments to $\vec x$, and that moreover, $f(\vec x)$ is the unique polynomial with that property. The uniqueness of $f$ and the fact that its degree is unbounded already proves that the bounded-degree IPS cannot refute $\Ff_{\text{sum}}(n,\bbF,\beta)$ for all $n \in \bbN$. To obtain the polynomial-size symmetric refutation, two things remain to be shown. Firstly, that the elementary symmetric polynomials are computable by $\Sym_n$-symmetric polynomial-size circuits; secondly, we need to work out how to combine the Boolean axioms with 
	$f(\vec x) \cdot \Big(\sum_{i=1}^n x_i - \beta \Big)$ in order to obtain $1$. The efficient symmetric computability of the $S_{n,k}$ is shown in \cite{shpilka2001depth}. 
	Regarding the Boolean axioms, we know by \cite[Proposition B.1]{lowerbounds_IPS} that $f(\vec x) \cdot \Big(\sum_{i=1}^n x_i - \beta \Big) = 1 + b(\vec x)$, where $b(\vec x)$ is a polynomial that can be written as $\sum_{i=1}^n (x_i^2 -x_i) \cdot p_i(\vec x)$, for appropriate polynomials $p_i$. These are the ones we need to determine. 
	We claim that the right choice is
	\begin{align*}
		p_i :=  \Big(\sum_{k=1}^n -\frac{k!}{\prod_{j = 0}^k (\beta - j)} S_{n \setminus i,k-1} \Big),
	\end{align*}	
	where \[
	S_{n \setminus i,k} = \sum_{\stackrel{S \subseteq [n] \setminus \{i\}}{|S| = k}} \prod_{j \in S} x_j.
	\]
	Now we verify that for this choice of $p_i$, we get $\sum_{i=1}^n (x_i^2 -x_i) \cdot p_i(\vec x) = b(\vec x)$.
	\begin{align*}
		\sum_{i=1}^n (x_i^2 -x_i) \cdot p_i(\vec x) &= \sum_{i \in [n]}\sum_{k=1}^n - \frac{k!}{\prod_{j = 0}^k (\beta - j)} S_{n,k,i^2}\\
		&+ \sum_{k=1}^n \frac{k \cdot k!}{\prod_{j = 0}^k (\beta - j)} S_{n,k}, \tag{$\star$}
	\end{align*}	
	where 
	\[
	S_{n,k,i^2} = \sum_{\stackrel{S \subseteq [n]}{|S| = k, i \in S}} x_i^2 \prod_{j \in S \setminus \{i\}} x_j.
	\]
	To see this, note that every degree-$k$ linear monomial $m$ in $\sum_{i=1}^n (x_i^2 -x_i) \cdot p_i(\vec x)$ is obtained in $k$ different ways, for exactly the $k$ choices of $i \in [n]$ such that $x_i$ appears in $m$. Now we compute $b(\vec x)$ and see that it matches:
	\begin{align*}
		\Big(\sum_{i=1}^n x_i - \beta \Big) \cdot f(\vec x) &=  \Big(\sum_{i=1}^n x_i - \beta \Big) \cdot \Big(- \sum_{k=0}^n \frac{k!}{\prod_{j=0}^k (\beta -j)} S_{n,k}\Big)\\	 
		&= \sum_{i \in [n]}\sum_{k=1}^n - \frac{k!}{\prod_{j = 0}^k (\beta - j)} S_{n,k,i^2}\\
		&+ \sum_{k=1}^{n} -\frac{k \cdot (k-1)!}{\prod_{j = 0}^{k-1} (\beta - j)} S_{n,k} + \sum_{k=0}^n \frac{\beta \cdot k!}{\prod_{j = 0}^{k} (\beta - j)} S_{n,k} 
	\end{align*}	
	Now combining the coefficients in the last line, we get
	\[
	-\frac{k \cdot (k-1)!}{\prod_{j = 0}^{k-1} (\beta - j)} + \frac{\beta \cdot k!}{\prod_{j = 0}^{k} (\beta - j)}  =  \frac{k \cdot k!}{\prod_{j = 0}^k (\beta - j)}. 
	\]
	Therefore, in total we have the following result which matches $(\star)$.
	\begin{align*}
		\Big(\sum_{i=1}^n x_i - \beta \Big) \cdot f(\vec x) &= 1+ \sum_{i \in [n]}\sum_{k=1}^n - \frac{k!}{\prod_{j = 0}^k (\beta - j)} S_{n,k,i^2} + \sum_{k=1}^n \frac{k \cdot k!}{\prod_{j = 0}^k (\beta - j)} S_{n,k}.
	\end{align*}	
	Hence, $p_i(\vec x)$ as defined above is correct. Let $\Stab(i)$ denote the pointwise stabiliser of $i$ in $\Sym_n$. There exists a $\Stab(i)$-symmetric polynomial-size circuit for computing $p_i(\vec x)$ because $S_{n \setminus i,k}$ is computable by a $\Stab(i)$-symmetric polynomial-size circuit: By renaming of variables, we can view $S_{n \setminus i,k}$ as $S_{n-1,k}$, and we know by \cite{shpilka2001depth} that this has an efficient $\Sym_{n-1}$-symmetric circuit. So in total, $f(\vec x)$, and $p_i(\vec x)$, for every $i \in [n]$, are the coefficients of the $\vec y$-variables in a $\Sym_n$-symmetric polynomial-size$\symIPSlin{}$ refutation.\\
	
	To construct a refutation for $\Ff_{\text{sum}(\vec x\vec y)}(n,\bbF,\beta)$, let us rename the variables of $f(\vec x)$ to $z_1,...,z_n$. Then $f(\vec z) \cdot \Big(\sum_{i \in [n]} z_i - \beta \Big) = 1+b(\vec z)$, where $b$ is the polynomial from above. Substituting $z_i \mapsto x_iy_i$, we get that $f(x_1y_1,...,x_ny_n) \cdot \Big(\sum_{i \in [n]} x_iy_i - \beta \Big) = 1+b(x_1y_1,...,x_ny_n)$. By what we showed above, $b(x_1y_1,...,x_ny_n) = \sum_{i=1}^n (x_i^2y_i^2-x_iy_i) \cdot p_i(x_1y_1,...,x_ny_n)$. Thus, our refutation uses the symmetric circuits we have constructed for $f(\vec z)$ and for the $p_i(\vec z)$, where we simply replace every input gate $z_i$ with the product $x_iy_i$. This is still symmetric with respect to the simultaneous action of $\Sym_n$ on $\vec x\vec y$. 
	It just remains to compute the polynomials $(x_i^2y_i^2-x_iy_i)$ used in $b(x_1y_1,...,x_ny_n)$ from the Boolean axioms $(x_i^2-x_i)$ and $(y_i^2-y_i)$. The following equation expresses $(x_i^2y_i^2-x_iy_i)$ as a $\Stab(i)$-symmetric circuit in the Boolean axioms $(x_i^2-x_i)$ and $(y_i^2-y_i)$.
	\begin{align*}
		(x_i^2y_i^2-x_iy_i) = (x_i^2-x_i)y_i^2 + (y_i^2-y_i)x_i^2-(x_i^2-x_i)(y_i^2-y_i)
	\end{align*}	
	For each $i \in [n]$, let $B_i^x$ denote the variable of the IPS certificate that stands for the Boolean axiom $(x_i^2-x_i)$, and let $B_i^y$ be the variable for  $(y_i^2-y_i)$ (given that the usual convention of using $y$ for the axiom variables would clash with the actual $\vec y$-variables of the instance here).
	For each $i \in [n]$, the refutation contains $B_i^x \cdot (2y_i^2 - y_i) \cdot p_i(x_1y_1,...,x_ny_n) + B_i^y \cdot x_i^2  \cdot p_i(x_1y_1,...,x_ny_n)$. This is linear in the $B_i$-variables, and yields in total a certificate that is symmetric with respect to the simultaneous action of $\Sym_n$ on $\vec x$ and $\vec y$ (note that it is not symmetric under exchanging $x_i$ and $y_i$, but our symmetry group does not allow this action).
	
	Finally, the fact that $\Ff_{\text{sum}(\vec x\vec y)}(n,\bbF,\beta)$ admits no constant-degree refutations follows e.g.\ from \cite[Corollary 5.9]{lowerbounds_IPS} which yields an exponential lower bound on depth-3 powering formulas -- a constant-degree refutation would be expressible as a polynomial-size depth-3 powering formula (just writing it as a sum of monomials would be an example of such a formula).
\end{proof}

\subsection{Pigeonhole principle}

\PHPupperBound*

Fix a number $n \in \bbN$ of holes.
Recall that for every set $D \subseteq [n+1]$ of pigeons, we define the polynomial
\[
B_D(\vec x) \coloneqq \sum_{\gamma \colon D \hookrightarrow [n]} \prod_{i \in D} x_{i\gamma(i)}.
\]
This can be seen as the sum over all injective mappings of the pigeons in $D$ to holes in $[n]$.
To prove Theorem \ref{thm:upperboundSymPHP}, we first describe a way to compute the polynomials $B_D$ for all $D \subseteq [n+1]$, and then show how to combine the $B_D$ to obtain a refutation of $\PHP(n+1,n)$.

For the computation of $B_D$, we also need the following polynomial for every $D \subseteq [n+1]$ and every $j \in [n]$: 
\[
X_{D \mapsto j}(\vec x) := \prod_{i \in D} x_{ij}.        
\]
This is obviously computable with a single multiplication gate. The resulting circuit is $\Stab(D \cup \{j\})$-symmetric, i.e.\ it is symmetric under every $\pi \in \Sym_{n+1} \times \Sym_n$ that fixes the hole $j$, and the pigeons $D$ setwise. 
The circuits for $B_D$ are defined recursively.
For $|D| = 1$, it is obvious how to compute $B_D$ with a single summation gate. For $D = \emptyset$, we let $B_D = 1$.
Both require just a single arithmetic gate. 
We now inductively construct circuits $C_D$ for $B_D$, for $|D| > 1$, assuming the circuits $C_D$ for all smaller $D$, and the circuits for all polynomials $X_{D \mapsto j}$ for all $D \subseteq I$ and $j \in J$ are available. The circuit $C_D$ performs the following computations:
\begin{align*}
	C_D := \sum_{k=1}^{|D|} (-1)^{k-1} \cdot \frac{k!}{|D|} \cdot &\Big(\sum_{\stackrel{D_k \subseteq D}{|D_k| = k}} \Big(\sum_{j \in [n]} X_{D_k \mapsto j}\Big) \cdot B_{D \setminus D_k} \Big).
\end{align*}	
For $B_{D \setminus D_k}$ and $X_{D_k \mapsto j}$, we use the circuits that are already constructed.
\begin{lemma}
	\label{lem:bijectionCircuitCorrect}
	For every $D \subseteq [n+1]$, $C_D(\vec x) = B_D$.
\end{lemma}	
\begin{proof}
	If $|D| \leq 1$, the construction of $C_D$ is non-recursive and ensures this. So consider a $D$ with $|D| > 1$. By induction on $m$, we prove for every $1 \leq m < |D|$:
	\begin{align*}
		\sum_{k=1}^{m} (-1)^{k-1} \cdot \frac{k!}{|D|} \cdot &\Big(\sum_{\stackrel{D_k \subseteq D}{|D_k| = k}} \Big(\sum_{j \in [n]} X_{D_k \mapsto j}\Big) \cdot B_{D \setminus D_k} \Big)\\
		=& B_D + (-1)^{m-1} \cdot \frac{(m+1)!}{|D|} \cdot \Big(\sum_{\stackrel{D_{m+1} \subseteq D}{|D_{m+1}| = m+1}} \Big(\sum_{j \in [n]} X_{D_{m+1} \mapsto j}\Big) \cdot B^{\setminus j}_{D \setminus D_{m+1}} \Big) \tag{$\star$}
	\end{align*}	
	Here, $B^{\setminus j}_{D \setminus D_{m+1}}$ denotes the sum over all monomials $m$ that encode a local bijection whose domain is $D \setminus D_{m+1}$ and whose image is a subset of $[n] \setminus \{j\}$.\\
	For $m = 1$, the expression computes for every pigeon $i \in D$ and every hole $j \in [n]$ the polynomial $B_{D \setminus \{i\}} \cdot x_{ij}$. Every monomial in $B_D$ is computed $|D|$ many times in this way (as the lift of each of its degree-$(|D|-1)$ submonomials). Thus, the factor $\frac{1}{|D|}$ ensures that we get exactly $B_D$. In addition to that, all degree-$|D|$ monomials describing a mapping that puts two pigeons in the same hole and is bijective everywhere else are computed. Each of these is obtained in exactly two distinct ways: Let $m$ be a degree-$|D|$ monomial involving the product $x_{ij}x_{i'j}$, and coding a local bijection between $|D|-2$ other pigeons and holes. Then $m$ appears in the summand with $D_k = \{i\}$ and also in the summand with $D_k = \{i'\}$. So $(\star)$ is true for $m = 1$.\\
	To prove $(\star)$ for $m+1$, we consider the value of the summand for $k = m+1$. This is
	\begin{align*}
		(-1)^{m+1-1} &\cdot \frac{(m+1)!}{|D|} \cdot \Big(\sum_{\stackrel{D_{m+1} \subseteq D}{|D_{m+1}| = m+1}} \Big(\sum_{j \in [n]} X_{D_{m+1} \mapsto j}\Big) \cdot B_{D \setminus D_{m+1}} \Big)\\
		=& 	(-1)^{m} \cdot \frac{(m+1)!}{|D|} \cdot \Big(\sum_{\stackrel{D_{m+1} \subseteq D}{|D_{m+1}| = m+1}} \Big(\sum_{j \in [n]} X_{D_{m+1} \mapsto j}\Big) \cdot B^{\setminus j}_{D \setminus D_{m+1}} \Big)\\
		&+ (-1)^{m} \cdot \frac{(m+2) \cdot (m+1)!}{|D|} \cdot \Big(\sum_{\stackrel{D_{m+2} \subseteq D}{|D_{m+2}| = m+2}} \Big(\sum_{j \in [n]} X_{D_{m+2} \mapsto j}\Big) \cdot B^{\setminus j}_{D \setminus D_{m+2}} \Big)
	\end{align*}	
	The second line exactly cancels the expression we get from the inductive hypothesis for $m$, except for $B_D$. The third line gives us $(\star)$ for $m+1$, together with $B_D$ that remains from the inductive hypothesis. The equality stated above holds because in
	\[
	\sum_{\stackrel{D_{m+1} \subseteq D}{|D_{m+1}| = m+1}} \Big(\sum_{j \in [n]} X_{D_{m+1} \mapsto j}\Big) \cdot B_{D \setminus D_{m+1}} \Big),
	\]
	every degree-$|D|$ monomial that puts $m+1$ many pigeons in the same hole $j$ and codes a bijection everywhere else appears exactly once. In addition to this, we also get every monomial that puts $m+2$ pigeons in the same hole $j$ and is bijective everywhere else; each such monomial is obtained in $m+2$ different ways. This proves $(\star)$ for every $m < |D|$.
	In total, we get
	\begin{align*}
		\sum_{k=1}^{|D|-1} (-1)^{k-1} \cdot \frac{k!}{|D|} \cdot &\Big(\sum_{\stackrel{D_k \subseteq D}{|D_k| = k}} \Big(\sum_{j \in [n]} X_{D_k \mapsto j}\Big) \cdot B_{D \setminus D_k} \Big)\\
		= &B_D + (-1)^{|D|-2} \cdot \frac{|D|!}{|D|} \cdot \Big(\sum_{j \in [n]} X_{D \mapsto j}\Big)  
	\end{align*}	   
	The summand for $k = |D|$ in $C_D$ exactly cancels $(-1)^{|D|-2} \cdot \frac{|D|!}{|D|} \cdot \Big(\sum_{j \in [n]} X_{D \mapsto j}\Big)$, so we obtain $B_D$ in the end.
\end{proof}	

\begin{lemma}
	\label{lem:computingAllBijections}
	There exists a circuit $C$ of size $\Oo(3^n)$ which is $(\Sym_{n+1} \times \Sym_{n})$-symmetric and computes all polynomials $B_D(\vec x)$ in the sense that for every $D \subseteq [n+1]$, there is a gate $g_D$ in $C$ whose output is $B_D(\vec x)$.
\end{lemma}	
\begin{proof}
	Let $C$ be the circuit $C_D$ as defined above, for $D = [n+1]$. This contains a subcircuit $C_D$ for every $D \subseteq [n+1]$ as a subcircuit. To estimate its total size, note that for every $D \subseteq [n+1]$, the size of $C_D$ is $\Oo(2^{|D|})$ if we only count the gates which do not belong to the subcircuits for the polynomials $B_{D \setminus D_k}$ appearing in the definition of $C_D$. Thus in total, in order to compute all $B_D$ for all $D \subseteq [n+1]$, we need a circuit of size $2^n \cdot n + \sum_{D \subseteq [n+1]} 2^{|D|} \leq \Oo(3^n)$. The $2^n \cdot n$ comes from the $X_{D \mapsto j}$ for all $D \subseteq [n+1]$ and all $j \in [n]$.
	For the symmetries, we observe that by definition, every $C_D$ is $\Stab(D)$-symmetric. So the total circuit $C = C_{[n+1]}$ is $\Stab([n+1])$-symmetric, and the setwise stabiliser of $[n+1]$ in $\Sym_{n+1} \times \Sym_{n}$ is $\Sym_{n+1} \times \Sym_{n}$ itself.
\end{proof}	
\noindent
Using the constructed $C_D$, we can get a $\vec y$-linear refutation for $\PHP(n+1,n)$ as follows. Define
\begin{align*}
	\alpha(n,k) &\coloneqq  \frac{1}{\binom{n}{k} \cdot (n-k)}\\	
	q_i(\vec x) &\coloneqq -\frac{1}{n+1} +  \sum_{k=1}^n \sum_{\stackrel{D \subseteq [n+1] \setminus \{i\}}{|D| = k}} \Big(- \alpha(n+1,k) \cdot B_D \Big) &\text{ for every } i \in [n+1] \\
	q_{(i,i',j)}(\vec x) &\coloneqq \sum_{D \subseteq I \setminus \{i,i'\}} 2 \alpha(n+1,|D|-1) \cdot B^{\setminus j}_{D} &\text{ for every } j \in [n], i \neq i' \in [n+1]
\end{align*}
As before, $B^{\setminus j}_D$ denotes the sum over all monomials encoding a local bijection from pigeons $D$ to a subset of $[n] \setminus \{j\}$. 
\begin{lemma}
	\begin{align*}
		C(\vec x, \vec y) := \sum_{i \in [n+1]} y_i \cdot q_i(\vec x) + \sum_{\stackrel{i \neq i' \in [n+1]}{j \in [n]}} y_{(i,i',j)} \cdot q_{(i,i',j)}(\vec x)
	\end{align*}	
	is a $\vec y$-linear refutation of $\PHP(n+1,n)$.
\end{lemma}	
\begin{proof}
	We have to prove $C(\vec x, \PHP(n+1,n)) = 1$. We start by computing the first sum:
	\begin{align*}
		\sum_{i \in [n+1]} \Big(&\sum_{j \in [n]} x_{ij} - 1 \Big) \cdot q_i(\vec x)\\
		= 1 &+ \sum_{k = 1}^{n+1} \Big( \sum_{\stackrel{D \subseteq [n+1]}{|D| = k}} \Big( B_D \cdot ( k \cdot (- \alpha(n+1,k-1)) + (n+1-k) \cdot \alpha(n+1,k) )\\
		&+ \sum_{\stackrel{D_2 \subseteq D}{|D_2| = 2}} \sum_{j \in [n]} -2 \cdot \alpha(n+1,k-1) \cdot X_{D_2 \mapsto j} \cdot B^{\setminus j}_{D \setminus D_2} \Big) \Big).		
	\end{align*}	
	In words: Every degree-$k$ monomial $m$ encoding a local bijection is obtained in $k$ different ways as the product of a variable $x_{ij}$ with a degree-$(k-1)$ submonomial of $m$. Moreover, $m$ is obtained in $(n+1-k)$ many ways in the product $(-1) \cdot B^{\setminus i}_D$, for every $i$ that does not appear in $m$. 
	In addition to these monomials, there appear all monomials encoding a local ``almost-bijection'', where one hole $j$ is occupied by two pigeons $i, i'$. Each of these monomials is created in two ways: Once it appears in $x_{ij} \cdot B^{\setminus i}_D$, and another time in $x_{i'j} \cdot B^{\setminus i'}_D$. Now a straightforward calculation shows that the coefficient $( k \cdot (- \alpha(n+1,k-1)) + (n+1-k) \cdot \alpha(n+1,k) )$ of $B_D$ is in fact zero. The remaining expression 
	\[
	\sum_{D \subseteq I}   \sum_{\stackrel{D_2 \subseteq D}{|D_2| = 2}} \sum_{j \in [n]} -2 \cdot \alpha(n+1,|D|-1) \cdot X_{D_2 \mapsto j} \cdot B^{\setminus j}_{D \setminus D_2} 
	\]
	is exactly cancelled by $\sum_{\stackrel{i \neq i' \in [n+1]}{j \in [n]}} x_{ij}x_{i'j} \cdot q_{(i,i',j)}(\vec x)$. 
	Hence, $C(\vec x, \PHP(n+1,n)) = 1$.
\end{proof}

\begin{lemma}
	\label{lem:circuitsizes}
	For every $i \in [n+1]$, $q_i(\vec x)$ can be computed by a $\Stab(i)$-symmetric circuit of size $\Oo(3^n)$. For every $j \in [n], i \neq i' \in [n+1]$, $q_{(i,i',j)}(\vec x)$ can be computed by a $\Stab(i,i',j)$-symmetric circuit of size $\Oo(3^n)$.
\end{lemma}	
\begin{proof}
	To compute $q_i$, we make use of the $(\Sym_{n+1} \times \Sym_[n])$-symmetric circuit from Lemma \ref{lem:computingAllBijections}, which contains a gate $g_D$ for every $D \subseteq [n+1]$, evaluating to $B_D$. Since $q_i$ only requires the polynomials $B_D$ for every $D \subseteq [n+1] \setminus \{i\}$, only the outputs of those gates $g_D$ are used, and the resulting circuit is symmetric under $\Stab(i) \leq \Sym_{n+1} \times \Sym_{n}$. The additional gates needed to compute $q_i$ are only linearly many that multiply the $B_D$ with the appropriate coefficients $-\alpha(n+1,|D|)$. 	
	In total, the circuit size is in $\Oo(3^n)$.\\
	Computing $q_{(i,i',j)}$ works similarly. The main difference is that in the circuit from Lemma \ref{lem:computingAllBijections}, we set every input variable mentioning hole $j$ to zero. This allows us to compute the polynomials $B_D^{\setminus j}$ needed for $q_{(i,i',j)}$. 
\end{proof}	

Theorem \ref{thm:upperboundSymPHP} now follows. The total refutation size $\Oo(n \cdot 3^n)$ is due to the fact that we need $n+1$ versions of the circuit from Lemma \ref{lem:computingAllBijections}: The original one, and one version for each $j \in [n]$ that is set to zero in the proof of Lemma \ref{lem:circuitsizes}. Moreover, it is easy to check that the entire certificate $C(\vec x, \vec y)$ is $(\Sym_{n+1} \times \Sym_{n})$-symmetric, and so is the refutation we obtain by putting together the circuits from Lemma \ref{lem:circuitsizes}.

\end{document}